\documentclass[letterpaper, 12pt]{article}[2000/05/19]
\usepackage[english]{babel}
\usepackage{amsfonts,amsmath,amssymb,amsthm,latexsym,amscd,mathrsfs}
\usepackage{ifthen,cite}
\usepackage[bookmarksnumbered=true]{hyperref}

\hypersetup{pdfpagetransition={Split}}

\newcommand{\evenhead}{Author \ name}
\newcommand{\oddhead}{Article \ name}
\newcommand{\theArticleName}{Article name}

\newcommand{\FirstPageHeading}[1]{\thispagestyle{empty}%
\noindent\raisebox{0pt}[0pt][0pt]{\makebox[\textwidth]{\protect\footnotesize \sf }}\par}

\newcommand{\ArticleName}[1]{\renewcommand{\theArticleName}{#1}\vspace{-2mm}\par\noindent {\LARGE\bf  #1\par}}
\newcommand{\Author}[1]{\vspace{5mm}\par\noindent {\Large  #1\par} \par\vspace{2mm}\par}
\newcommand{\Address}[1]{\vspace{2mm}\par\noindent {\it #1} \par}
\newcommand{\Email}[1]{\ifthenelse{\equal{#1}{}}{}{\par\noindent {\rm E-mail: }{\it  #1} \par}}
\newcommand{\URLaddress}[1]{\ifthenelse{\equal{#1}{}}{}{\par\noindent {\rm URL: }{\tt  #1} \par}}
\newcommand{\EmailD}[1]{\ifthenelse{\equal{#1}{}}{}{\par\noindent {$\phantom{\dag}$~\rm E-mail: }{\it  #1} \par}}
\newcommand{\URLaddressD}[1]{\ifthenelse{\equal{#1}{}}{}{\par\noindent {$\phantom{\dag}$~\rm URL: }{\tt  #1} \par}}

\newcommand{\Abstract}[1]{\vspace{6mm}\par\noindent\hspace*{10mm}
\parbox{140mm}{\small {\bf Abstract.} #1}\par}
\newcommand{\Keywords}[1]{\vspace{3mm}\par\noindent\hspace*{10mm}
\parbox{140mm}{\small {\bf Key words:} \rm #1}\par}
\newcommand{\Classification}[1]{\vspace{3mm}\par\noindent\hspace*{10mm}
\parbox{140mm}{\small {\it 2000 Mathematics Subject Classification:} \rm #1}\vspace{3mm}\par}
\newcommand{\ShortArticleName}[1]{\renewcommand{\oddhead}{#1}}
\newcommand{\AuthorNameForHeading}[1]{\renewcommand{\evenhead}{#1}}

\long\def\@makecaption#1#2{%\vskip\abovecaptionskip
  \sbox\@tempboxa{\small \textbf{#1.}\ \ #2}%
  \ifdim \wd\@tempboxa >\hsize
    {\small \textbf{#1.}\ \ #2}\par \else
    \global \@minipagefalse
    \hb@xt@\hsize{\hfil\box\@tempboxa\hfil}%
  \fi \vskip\belowcaptionskip}
\def\numberwithin#1#2{\@ifundefined{c@#1}{\@nocounterr{#1}}{%
  \@ifundefined{c@#2}{\@nocnterr{#2}}{%
  \@addtoreset{#1}{#2}%
  \toks@\@xp\@xp\@xp{\csname the#1\endcsname}%
  \@xp\xdef\csname the#1\endcsname
    {\@xp\@nx\csname the#2\endcsname.\the\toks@}}}}
\def\E^#1{{\buildrel #1 \over\vee}}
\newtheorem{theorem}{Theorem}

\newtheorem{proposition}{Proposition}
{\theoremstyle{definition}

}

\begin{document}

\FirstPageHeading{V.I. Gerasimenko and D.O. Polishchuk}

\ShortArticleName{Dynamics of Correlations}

\AuthorNameForHeading{V.I. Gerasimenko and D.O. Polishchuk}

\ArticleName{Dynamics of Correlations \\of Bose and Fermi Particles}

\Author{V.I. Gerasimenko$^\ast$ and D.O. Polishchuk$^\ast$$^\ast$}

\Address{$^\ast$\hspace*{1mm}Institute of Mathematics of NAS of Ukraine,\\
    \hspace*{3mm}3, Tereshchenkivs'ka Str.,\\
    \hspace*{3mm}01601 Kyiv-4, Ukraine}
\EmailD{gerasym@imath.kiev.ua}

\Address{$^\ast$$^\ast$Taras Shevchenko National University of Kyiv,\\
    \hspace*{4mm}Department of Mechanics and Mathematics,\\
    \hspace*{4mm}2, Academician Glushkov Av.,\\
    \hspace*{4mm}03187 Kyiv, Ukraine}
\EmailD{polishuk.denis@gmail.com}

\bigskip

\Abstract{
We discuss the origin of the microscopic description of correlations in quantum many-particle systems obeying
Fermi-Dirac and Bose-Einstein statistics.
For correlation operators that give the alternative description of the quantum state evolution of Bose and Fermi particles,
we deduce the von Neumann hierarchy of nonlinear equations and
construct the solution of its initial-value problem in the corresponding spaces of sequences of trace class operators.
The links of constructed solution both with the solution of the quantum BBGKY hierarchy
and with the nonlinear BBGKY hierarchy for marginal correlation operators are discussed.
The solutions of the Cauchy problems of these hierarchies are constructed for initial data satisfying a chaos property.}

\bigskip

\vspace{2pc}
\Classification{35Q40; 47D06.}

\makeatletter
\renewcommand{\@evenhead}{
\hspace*{-3pt}\raisebox{-15pt}[\headheight][0pt]{\vbox{\hbox to \textwidth {\thepage \hfil \evenhead}\vskip4pt \hrule}}}
\renewcommand{\@oddhead}{
\hspace*{-3pt}\raisebox{-15pt}[\headheight][0pt]{\vbox{\hbox to \textwidth {\oddhead \hfil \thepage}\vskip4pt\hrule}}}
\renewcommand{\@evenfoot}{}
\renewcommand{\@oddfoot}{}
\makeatother

\newpage
\protect\tableofcontents
%\newpage
\vspace{0.5cm}

\section{Introduction}

The paramount importance of the mathematical description of correlations in numerous problems of the modern statistical mechanics
is well-known.
We refer, for example, to such fundamental problems as quantum measurements and the rigorous derivation of quantum kinetic equations \cite{AGT,AA,BDGM,BGGM2,BCEP3,PP09,ESY,Sp80,Got}, in particular,
kinetic equations that describe Bose gases in the condensate state \cite{FL,ESchY3,ESchY2,M1,Ger}.

There are several approaches to the description of fluctuations in many-particle systems.
Among them we mention the well-known approach of dynamics of correlations \cite{Bol,Pri} and its applications in plasma theory.
Equilibrium states in framework of correlation operators are rigorously described in works \cite{Gen,Rul}.

The aim of the paper is to formulate rigorously the evolution equations that describe non-equilibrium correlations
of Fermi and Bose many-particle systems and to construct a solution of the corresponding Cauchy problem.
Then using the constructed solution to establish its links with the solutions of the BBGKY hierarchies for
marginal density and marginal correlation operators \cite{BogLect,Pe95,CGP97}.

The paper is organized in the following way.

In section 2 the description of evolution of states of Fermi and Bose many-particle systems in
framework of correlation operators is given. The correlation operators are defined by means of cluster expansions
of the density operators (density matrices), which are solutions of the initial-value problem of the von Neumann equations
of a non-fixed, i.e. arbitrary but finite, number of particles. We deduce the von Neumann hierarchy that governs the evolution
of correlation operators of Fermi and Bose particles with the general type of an interaction potential
and construct a solution of its initial-value problem,
in particular in case of initial data satisfying a chaos property.

In section 3 we give the rigorous justification of the derivation of evolution equations for the correlation operators.
We prove that the solution of the initial-value problem of the von Neumann hierarchy generates
a group of nonlinear operators of the class $C_{0}$ in the suitable spaces of sequences of trace class operators.
In these spaces the existence and uniqueness of a strong and a weak solution of the Cauchy problem of the von Neumann hierarchy is proved.

In section 4 using introduced correlation operators we define functionals of average values and a dispersion of observables.
Moreover, we give the description of states of Fermi and Bose many-particle systems in
framework of marginal density and marginal correlation operators, that are adopted for the description of infinite-particle systems.
The solutions of the corresponding BBGKY hierarchies for these operators are deduced on the basis of constructed solution of
the von Neumann hierarchy. We demonstrate that the concept of cumulants of groups of operators of the von Neumann equations
forms the basis for the solution expansions of mentioned hierarchies of evolution equations.
In the final section we give concluding comments and remarks.

\section{The von Neumann hierarchy for correlation operators of Fermi and Bose many-particle systems}

\subsection{Preliminary facts}
Let $\mathcal{H}$ be the one-particle Hilbert space, then the $n$-particle space
$\mathcal{H}_n\equiv\mathcal{H}^{\otimes n}$ is a tensor product of $n$ Hilbert spaces $\mathcal{H}$.
We adopt the usual convention that $\mathcal{H}^{\otimes 0}=\mathbb{C}$.
Let $\mathfrak{S}_{n}$ be the permutation group of the set $(1,\ldots,n)$. We assign to each permutation $\pi \in \mathfrak{S}_{n}$
an isomorphism $p_{\pi}$ of $\mathcal{H}^{\otimes n}$ onto itself. The operator $p_{\pi}$
transforms factorized elements $\psi_1\otimes \psi_2\otimes\ldots\otimes \psi_n \in \mathcal{H}^{\otimes n}$
into $\psi_{\pi(1)}\otimes \psi_{\pi(2)}\otimes\ldots\otimes \psi_{\pi(n)} \in \mathcal{H}^{\otimes n}$.
The symmetrization operator $\mathcal{S}^{+}_{n}$ and the anti-symmetrization operator
$\mathcal{S}^{-}_{n}$ on $\mathcal{H}^{\otimes n}$ are defined by the formula
\begin{eqnarray}\label{Sn}
    &&\mathcal{S}^{\pm}_{n}\doteq\frac {1}{n!}\sum\limits_{\pi\epsilon \mathfrak{S}_{n}}(\pm1)^{|\pi|}p_{\pi},
\end{eqnarray}
where $|\pi|$ denotes the number of transpositions in the permutation $\pi$.
The operators $\mathcal{S}^{\pm}_{n}$ are orthogonal projectors, i.e.
$({\mathcal{S}^{\pm}_{n}})^{2}=\mathcal{S}^{\pm}_{n}$,
ranges of which are correspondingly the symmetric tensor product $\mathcal{H}_{n}^{+}$
and the antisymmetric tensor product $\mathcal{H}_{n}^{-}$ of $n$ Hilbert spaces $\mathcal{H}$.
We denote by $\mathcal{F}^{\pm}_{\mathcal{H}}={\bigoplus\limits}_{n=0}^{\infty}\mathcal{H}_{n}^{\pm}$
the Fock spaces over Hilbert space  $\mathcal{H}$, that are associated with the corresponding Fermi and Bose systems \cite{BR,DauL_5}.

Hereinafter we consider a quantum system of a non-fixed,
i.e. arbitrary but finite, number of identical (spinless) particles with unit mass $m=1$ in the space
$\mathbb{R}^{\nu},$ $\nu\geq 1$, that obey Fermi-Dirac or Bose-Einstein statistics.
The Hamiltonian $H=\bigoplus_{n=0}^{\infty}H_{n}$ of the system  is a self-adjoint operator ($H_{0}=0$) with domain
$\mathcal{D}(H)=\{\psi=\oplus\psi_{n}\in{\mathcal{F}^{\pm}_{\mathcal{H}}}\mid \psi_{n}\in\mathcal{D}
(H_{n})\in\mathcal{H}^{\pm}_{n},\sum_{n}\|H_{n}\psi_{n}\|^{2}<\infty\}\subset{\mathcal{F}^{\pm}_{\mathcal{H}}}$.
On functions $\psi_n$ that belong to the subspaces $\mathcal{S}^{\pm}_{n}L^{2}_{0}(\mathbb{R}^{\nu n})\subset \mathcal{D}(H_n) \subset
\mathcal{S}^{\pm}_{n}L^{2}(\mathbb{R}^{\nu n})$ of infinitely
differentiable functions with compact supports $n$-particle Hamiltonian $H_{n}$ acts according to the formula
\begin{eqnarray}\label{H_Zag}
    &&H_{n}\psi_n = \sum\limits_{i=1}^{n}K(i)\psi_n
       +\sum\limits_{k=1}^{n}\,\sum\limits_{i_{1}<\ldots<i_{k}=1}^{n}\Phi^{(k)}(i_{1},\ldots,i_{k})\psi_{n},
\end{eqnarray}
where $K(i)\psi_n = -\frac{\hbar}{2}\Delta_{q_i}\psi_n$ is the operator of the kinetic energy, $2\pi\hbar$ is a Planck constant
and $\Phi^{(k)}(i_{1},\ldots,i_{k})\psi_{n}=\Phi^{(k)}(q_{i_{1}},\ldots,q_{i_{k}})\psi_{n}$ is a $k$-body interaction potential.
We assume that the functions $\Phi^{(k)}, k\geq1,$ are translation-invariant, symmetric with respect to permutations
of their arguments and satisfy the Kato conditions \cite{Kato}, which guarantee self-adjointness of operator (\ref{H_Zag}).

States of systems of Bose and Fermi particles belong to the corresponding spaces
$\mathfrak{L}^{1}(\mathcal{F}^{\pm}_\mathcal{H})= \oplus_{n=0}^{\infty}
\mathfrak{L}^{1}(\mathcal{H}^{\pm}_{n})$ of sequences
$f=(f_0,f_{1},\ldots,f_{n},\ldots)$ of trace class operators
$f_{n}\equiv f_{n}(1,\ldots,n)\in\mathfrak{L}^{1}(\mathcal{H}^{\pm}_{n})$ and $f_0 \in \mathbb{C}$,
that satisfy the symmetry condition:
$f_{n}(1,\ldots,n)=f_{n}(i_{1},\ldots,i_{n})$ for arbitrary $(i_{1},\ldots,i_{n})\in(1,\ldots,n)$,
equipped with the norm
\begin{eqnarray*}
    &&\|f\|_{\mathfrak{L}^{1} (\mathcal{F}^{\pm}_\mathcal{H})}=
        \sum\limits_{n=0}^{\infty} \|f_{n}\|_{\mathfrak{L}^{1}(\mathcal{H}^{\pm}_{n})}=
        \sum\limits_{n=0}^{\infty}~\mathrm{Tr}_{1,\ldots,n}|f_{n}(1,\ldots,n)|,
\end{eqnarray*}
where $\mathrm{Tr}_{1,\ldots,n}$ are partial traces over $1,\ldots,n$ particles \cite{DauL_5,Pe95}.
We denote by $\mathfrak{L}^{1}_{0}$ the everywhere dense set in
$\mathfrak{L}^{1}(\mathcal{F}^{\pm}_\mathcal{H})$ of finite sequences
of degenerate operators \cite{Kato} with infinitely differentiable kernels with compact supports.
Note that the space $\mathfrak{L}^{1}(\mathcal{F}^{\pm}_\mathcal{H})$ contains more general sequences of operators
than those that determine states of systems.

The Bose-Einstein and Fermi-Dirac statistics endows the states with additional symmetry properties.
We illustrate them on kernels of operators \cite{Pe95}.
Let $f_{n}(q_{1},\ldots,$ $q_{n};q'_{1},\ldots,q'_{n})$ is the kernel of the operator $f_{n}\in\mathfrak{L}^{1}(\mathcal{S}^{\pm}_{n}L^{2}(\mathbb{R}^{\nu n}))$.
In case of the Bose-Einstein statistics the kernel is a symmetric function
with respect to permutations in each group of arguments
\begin{eqnarray*}
  &&f_{n}(q_{1},q_{2},\ldots,q_{n};q'_{1},q'_{2},\ldots,q'_{n})=
       f_{n}(q_{\pi(1)},q_{\pi(2)},\ldots,q_{\pi(n)};q'_{\pi'(1)},q'_{\pi'(2)},\ldots,q'_{\pi'(n)}),
\end{eqnarray*}
and in case of the Fermi-Dirac statistics the corresponding kernel is an antisymmetric function
\begin{eqnarray*}
  &&f_{n}(q_{1},q_{2},\ldots,q_{n};q'_{1},q'_{2},\ldots,q'_{n})=(-1)^{(|\pi|+|\pi'|)}
       f_{n}(q_{\pi(1)},q_{\pi(2)},\ldots,q_{\pi(n)};q'_{\pi'(1)},q'_{\pi'(2)},\ldots,q'_{\pi'(n)}),
\end{eqnarray*}
where $\pi,\pi'\in \mathfrak{S}_{n}$.

The evolution of all possible states,
i.e. sequences $D(t)=(1,D_{1}(t,1),\ldots,D_{n}(t,1,\ldots,n),$ $\ldots)\in\mathfrak{L}^{1}(\mathcal{F}^{\pm}_\mathcal{H})$
of the density operators $D_{n}(t), n\geq1$, is described by the initial-value
problem of a sequence of the von Neumann equations (the quantum Liouville equations) \cite{vN55,BogLect,Pe95}
\begin{eqnarray}\label{vonNeumannEqn}
     &&\frac{d}{d t}D(t)=-\mathcal{N}D(t),\\
  \label{F-N12}
     &&D(t)|_{t=0}=D(0),
\end{eqnarray}
where for $f\in \mathfrak{L}^{1}_{0}(\mathcal{F}^{\pm}_\mathcal{H})\subset\mathcal{D}(\mathcal{N})\subset \mathfrak{L}^{1}(\mathcal{F}^{\pm}_\mathcal{H})$
the generator $\mathcal{N}=\oplus^{\infty}_{n=0}\mathcal{N}_{n}$ of the von Neumann equation is defined by
\begin{eqnarray}\label{komyt}
   &&(\mathcal{N}f)_{n}=-\frac{i}{\hbar}\big[f_{n},H_{n}\big]\doteq
       -\frac{i}{\hbar}(f_{n}H_{n}-H_{n}f_{n}),
\end{eqnarray}
where $H_{n}$ is the Hamiltonian (\ref{H_Zag}).

In the spaces of sequences of trace class operators $\mathfrak{L}^{1}(\mathcal{F}^{\pm}_\mathcal{H})$ for initial-value problem
(\ref{vonNeumannEqn})-(\ref{F-N12}) the following statement is true.
\begin{proposition}
A unique solution of initial-value problem (\ref{vonNeumannEqn})-(\ref{F-N12}) is determined by the formula
\begin{eqnarray}\label{rozv_fon-N}
    && D(t)=\mathcal{G}(-t)D(0),
\end{eqnarray}
where $\mathcal{G}(-t)=\oplus_{n=0}^{\infty}\mathcal{G}_{n}(-t)$,
\begin{eqnarray}\label{groupG}
    &&\mathcal{G}_{n}(-t)f_{n}\doteq e^{-{\frac{i}{\hbar}}tH_{n}}f_{n}e^{{\frac{i}{\hbar}}tH_{n}}.
\end{eqnarray}
For $D(0)\in\mathfrak{L}^{1}_{0}(\mathcal{F}^{\pm}_\mathcal{H})\subset
\mathfrak{L}^{1}(\mathcal{F}^{\pm}_\mathcal{H})$ it is a strong (classical) solution and for arbitrary
$D(0) \in \mathfrak{L}^{1}(\mathcal{F}^{\pm}_\mathcal{H})$ it is a weak (generalized) solution.
\end{proposition}
For the case of the Maxwell-Boltzmann statistics an analogue of this proposition is proved in \cite{BR},\cite{DauL_5}.

In the spaces $\mathfrak{L}^{1}(\mathcal{F}^{\pm}_\mathcal{H})$
the mapping $: t\rightarrow\mathcal{G}(-t)f$
is an isometric strongly continuous group that preserves positivity and self-adjointness of operators.

For $f\in\mathfrak{L}_{0}^{1}(\mathcal{F}^{\pm}_\mathcal{H})\subset\mathcal{D}(-\mathcal{N})$
in the sense of the norm convergence in the spaces $\mathfrak{L}^{1}(\mathcal{F}^{\pm}_\mathcal{H})$ there exists the limit by which
the infinitesimal generator of the group of evolution operators (\ref{groupG})
is determined:
\begin{eqnarray}\label{infOper}
    &&\lim\limits_{t\rightarrow 0}\frac{1}{t}(\mathcal{G}(-t)f-f)=-\frac{i}{\hbar}(Hf-fH)\doteq-\mathcal{N}f,
\end{eqnarray}
where $H=\oplus^{\infty}_{n=0}H_{n}$ is the Hamiltonian (\ref{H_Zag}) and
the operator $(-\frac{i}{\hbar}(Hf-fH))$ is defined in the domain $\mathcal{D}(H)\subset\mathcal{F}^{\pm}_\mathcal{H}.$

Symmetrization and anti-symmetrization operators (\ref{Sn}) are integrals of motion of the von Neumann equation (\ref{vonNeumannEqn}),
and as a consequence
\begin{eqnarray*}
    &&\mathcal{G}_{n}(-t)\mathcal{S}^{\pm}_{n}=\mathcal{S}^{\pm}_{n}\mathcal{G}_{n}(-t),
\end{eqnarray*}
and
\begin{eqnarray*}
    &&\mathcal{N}_{n}\mathcal{S}^{\pm}_{n}=\mathcal{S}^{\pm}_{n}\mathcal{N}_{n},
\end{eqnarray*}
i.e. the symmetry of states is preserved within evolutionary process.

\subsection{Cluster expansions of density operators}
We introduce operators
$g(t)=(0,g_{1}(t,1),\ldots,g_{s}(t,1,\ldots,s),\ldots)\in \mathfrak{L}^{1}(\mathcal{F}^{\pm}_\mathcal{H})$
that are defined by the cluster expansions of the density operators
$D(t)=(1,D_{1}(t,1),\ldots,D_{s}(t,1,\ldots,s),\ldots)\in \mathfrak{L}^{1}(\mathcal{F}^{\pm}_\mathcal{H})$
\begin{eqnarray}\label{D_(g)N}
    &&D_{s}(t,Y)= g_{s}(t,Y)+
        \sum\limits_{\mbox{\scriptsize $\begin{array}{c}\mathrm{P}:Y=\bigcup_{i}X_{i},\\|\mathrm{P}|>1 \end{array}$}}
        \mathcal{S}^{\pm}_{s}\prod_{X_i\subset \mathrm{P}}g_{|X_i|}(t,X_i),\quad s\geq1,
\end{eqnarray}
where ${\sum\limits}_{\mathrm{P}:Y=\bigcup_{i} X_{i},\,|\mathrm{P}|>1}$ is the sum over all possible partitions $\mathrm{P} $
of the set $Y\equiv(1,\ldots,s)$ into $|\mathrm{P}|>1$ nonempty mutually disjoint subsets $X_i\subset Y$.

We give instances of recursion relations (\ref{D_(g)N})
\begin{eqnarray*}
  &&D_{1}(t,1)=g_{1}(t,1),\\
  &&D_{2}(t,1,2)=g_{2}(t,1,2)+\mathcal{S}^{\pm}_{2}g_{1}(t,1)g_{1}(t,2).
\end{eqnarray*}

The sequence of operators $g(t)$ defined by recursion relations (\ref{D_(g)N})
has a meaning of a sequence of correlation operators of Fermi and Bose many-particle systems,
i.e. a sequence of operators that characterize the correlations of particle states\footnote{The motivation of introduction 
of these operators is considered in section 4.}.
We remark that cluster expansions (\ref{D_(g)N}) without operators $\mathcal{S}^{\pm}_{s}$ correspond
to correlation operators that describe states of systems obeying the Maxwell-Boltzmann statistics.

In order to construct the solution of recursion relations (\ref{D_(g)N}), i.e. to express the correlation operators
in terms of the density operators, we introduce the following notions.
On sequences of operators  $f,\widetilde{f}\in \mathfrak{L}^{1}(\mathcal{F}^{\pm}_\mathcal{H}) $ we define the $\circledast$-product
\begin{eqnarray}\label{Product}
    &&(f\circledast \widetilde{f})_{|Y|}(Y)=\sum\limits_{Z\subset Y} \mathcal{S}^{\pm}_{|Y|}\,f_{|Z|}(Z)
        \,\widetilde{f}_{|Y\backslash Z|}(Y\backslash Z),
\end{eqnarray}
where the operators $\mathcal{S}^{\pm}_{|Y|}$ are defined by formula (\ref{Sn}) and $\sum_{Z\subset Y}$
is the sum over all subsets $Z$ of the set $Y\equiv(1,\ldots,s)$. We remark that similar product was introduced in \cite{Gen,Rul}
for the studying of equilibrium states.

By means of definition (\ref{Product}) of the $\circledast$-product we introduce the mapping ${\mathbb E}\mathrm{xp}_{\circledast}$
and the inverse mapping ${\mathbb L}\mathrm{n}_{\circledast}$ on sequences $h=(0,h_1(1),\ldots,h_n(1,\ldots,n),\ldots)$
of operators  $h_n\in\mathfrak{L}^{1}(\mathcal{H}^{\pm}_{n})$ by the expansions
\begin{eqnarray}\label{circledExp}
   &&({\mathbb E} \mathrm{xp}_{\circledast}h )_{|Y|}(Y)=\big({\bf 1}+
      \sum\limits_{n=1}^{\infty} \frac{h^{\circledast n}}{n!}\big)_{|Y|}(Y)=\\
   &&=\delta_{|Y|,0}+\sum\limits_{\mathrm{P}:\,Y=\bigcup_{i}X_{i}}{\mathcal{S}^{\pm}_{s}}
      \prod_{X_i\subset \mathrm{P}}h_{|X_i|}(X_i),\nonumber
\end{eqnarray}
where ${\bf 1}=(1,0,\ldots,0,\ldots)$ and $\delta_{|Y|,0}$ is a Kronecker symbol,
and correspondingly,
\begin{eqnarray}\label{circledLn}
   &&({\mathbb L}\mathrm{n}_{\circledast} ({\bf 1}+h))_{|Y|}(Y)=
       \big(\sum\limits_{n=1}^{\infty} (-1)^{n-1}\,\frac{h^{\circledast n}}{n}\big)_{|Y|}(Y)=\\
   &&=\sum\limits_{\mathrm{P}:\,Y=\bigcup_{i}X_{i}}
       (-1)^{|\mathrm{P}|-1}(|\mathrm{P}|-1)!\,
       {\mathcal{S}^{\pm}_{s}}\prod_{X_i\subset \mathrm{P}}h_{|X_i|}(X_i).\nonumber
\end{eqnarray}
Then in terms of sequences of operators recursion relations (\ref{D_(g)N}) are rewritten in the form
\begin{eqnarray*}\label{DtoGcircledStar}
    &&D(t)={\mathbb E}\mathrm{xp}_{\circledast}\,\,g(t).
\end{eqnarray*}
From this equation we obtain
\begin{eqnarray*}
    &&g(t)={\mathbb L}\mathrm{n}_{\circledast}D(t),
\end{eqnarray*}
where $D(t)={\bf 1}+(0,D_1(1),D_2(1,2),\ldots,D_n(1,\ldots,n),\ldots)$.

Thus, according to definition (\ref{Product}) of the $\circledast$-product, 
in the component-wise form solutions of recursion relations (\ref{D_(g)N})
are represented by the expansions
\begin{eqnarray}\label{gfromDFB}
   &&g_{s}(t,Y)=D_{s}(t,Y)+\nonumber\\
   &&+\sum\limits_{\mbox{\scriptsize $\begin{array}{c}\mathrm{P}:Y=\bigcup_{i}X_{i},\\|\mathrm{P}|>1\end{array}$}}
       (-1)^{|\mathrm{P}|-1}(|\mathrm{P}|-1)!\,
       {\mathcal{S}^{\pm}_{s}}\prod_{X_i\subset \mathrm{P}}D_{|X_i|}(t,X_i), \quad s\geq1.
\end{eqnarray}
In terms of kernels of operators the first two equations of expansions (\ref{gfromDFB}) have the following form
\begin{eqnarray*}
  &&g_{1}(t,q_1;q_1')=D_{1}(t,q_1;q_1'),\\
  &&g_{2}(t,q_1,q_2;q_1',q_2')=D_{2}(t,q_1,q_2;q_1',q_2')-\\
  &&-\frac{1}{2!}(D_{1}(t,q_1;q_1')D_{1}(t,q_2;q_2')\pm D_{1}(t,q_2;q_1')D_{1}(t,q_1;q_2')).
\end{eqnarray*}
The structure of expansions (\ref{gfromDFB}) means that the correlation operators have a sense of cumulants (semi-invariants)
of the density operators governed by the Cauchy problem of the von Neumann equation (\ref{vonNeumannEqn})-(\ref{F-N12}).

Thus, correlation operators (\ref{gfromDFB}) give an alternative approach to the description of states
of the Fermi and Bose many-particle systems, namely in the framework of dynamics of correlations.

\subsection{The von Neumann hierarchy for correlation operators}
From expansions (\ref{gfromDFB}) using the von Neumann equation (\ref{vonNeumannEqn})
we derive the hierarchy of evolution equations for the correlation operators
$g(t)=(0,g_{1}(t,1),\ldots,g_{s}(t,1,\ldots,s),\ldots)\in \mathfrak{L}^{1}(\mathcal{F}^{\pm}_\mathcal{H})$
\begin{eqnarray}\label{vonNeumanHierarchyAlternateDescriptionFB}
    &&\frac{d}{dt}g_{s}(t,Y)=-\mathcal{N}_{s}(Y)g_{s}(t,Y)+\\
    &&\hskip-5mm+\sum\limits_{\mbox{\scriptsize $\begin{array}{c}\mathrm{P}:Y=\bigcup_{i}X_{i},\\|\mathrm{P}|>1\end{array}$}}\hskip-2mm
       \sum\limits_{\mbox{\scriptsize$\begin{array}{c}{Z_{1}\subset X_{1}},\\Z_{1}\neq\emptyset\end{array}$}}\hskip -2mm \ldots \hskip -2mm
       \sum\limits_{\mbox{\scriptsize $\begin{array}{c} {Z_{|\mathrm{P}|}\subset X_{|\mathrm{P}|}},\\Z_{|\mathrm{P}|}\neq\emptyset \end{array}$}}
       \hskip -2mm \big(-\mathcal{N}_{\mathrm{int}}^{(\sum\limits_{r=1}^{|\mathrm{P}|}|Z_{{r}}|)}
       (Z_{{1}},\ldots,Z_{{|\mathrm{P}|}})\big)\mathcal{S}^{\pm}_{s}\prod_{X_{i}\subset \mathrm{P}}g_{|X_{i}|}(t,X_{i}),\nonumber
\end{eqnarray}
where
\begin{eqnarray}\label{oper Nint2}
     &&\mathcal{N}^{(n)}_{\mathrm{int}}(1,\ldots,n)\doteq-\frac{i}{\hbar}\big[~\cdot~,\Phi^{(n)}(1,\ldots,n)\big],
\end{eqnarray}
$\Phi^{(n)}$ is a $n$-body interaction potential defined in (\ref{H_Zag}),
${\sum\limits}_{\mbox{\scriptsize $\begin{array}{c}\mathrm{P}:Y=\bigcup_{i} X_{i},\,|\mathrm{P}|>1\end{array}$}}$
is the sum over all possible partitions $\mathrm{P} $ of the set $Y\equiv(1,\ldots,s)$ into $|\mathrm{P}|>1$
nonempty mutually disjoint subsets $X_i\subset Y$ and $\sum_{Z_{j}\subset X_{j},\,Z_{j}\neq\emptyset}$
is a sum over nonempty subsets $Z_{j}\subset X_{j}$. We refer to the hierarchy of equations
(\ref{vonNeumanHierarchyAlternateDescriptionFB}) as the von Neumann hierarchy.

Hierarchy (\ref{vonNeumanHierarchyAlternateDescriptionFB}) together with the initial state defined by
\begin{eqnarray}\label{vonNeumanHierarchyAlternateDescriptionInitialValue}
     && g_{s}(t,Y)\big|_{t=0}=g_{s}(0,Y),\quad s\geq1,
\end{eqnarray}
describes the evolution of all possible states of Fermi and Bose many-particle systems
in framework of correlation operators.

We give an explanatory comment concerning the origin of
the von Neumann hierarchy (\ref{vonNeumanHierarchyAlternateDescriptionFB}).
As a result of the differentiation of the definition of correlation operators (\ref{gfromDFB}) by time variable
in the sense of pointwise convergence in the spaces $\mathfrak{L}^{1}(\mathcal{H}^{\pm}_{n})$
and using cluster expansions (\ref{D_(g)N}) we obtain
\begin{eqnarray*}
   &&\frac{d}{dt}g_{s}(t,Y)=
       \sum\limits_{\mbox{\scriptsize $\begin{array}{c}\mathrm{P}:Y=\bigcup_{i}X_{i}\end{array}$}}(-1)^{|\mathrm{P}|-1}(|\mathrm{P}|-1)!\,
       {\mathcal{S}^{\pm}_{s}}\sum\limits_{X_i\subset \mathrm{P}}(-\mathcal{N}_{|X_i|})(X_i)
       \prod_{X_i\subset \mathrm{P}}D_{|X_i|}(t,X_i)=\\
   &&=\hskip-2mm\sum\limits_{\mbox{\scriptsize $\begin{array}{c}\mathrm{P}:Y=\bigcup_{i}X_{i}\end{array}$}}
       \hskip -2mm(-1)^{|\mathrm{P}|-1}(|\mathrm{P}|-1)!\,
       {\mathcal{S}^{\pm}_{s}}\sum\limits_{X_i\subset \mathrm{P}}(-\mathcal{N}_{|X_i|})(X_i) \prod_{X_i\subset \mathrm{P}}\sum\limits_{\mbox{\scriptsize $\begin{array}{c} \mathrm{P'}:X_i=\bigcup_{j_i}Z_{j_i}\end{array}$}}
       \prod_{Z_{j_i}\subset \mathrm{P'}}g_{|Z_{j_i}|}(t,Z_{j_i}).
\end{eqnarray*}
Changing the summation indexes $\mathrm{P}$ and $\mathrm{P'}$ in this expression
and collecting similar terms gives us the chain of equalities (\ref{vonNeumanHierarchyAlternateDescriptionFB})
which we interpret as the hierarchy of evolution equations for the sequence of correlation operators\footnote{The rigorous
justification is given in section 3.}.

We cite an instance the typical equations of hierarchy (\ref{vonNeumanHierarchyAlternateDescriptionFB})
\begin{eqnarray*}\label{vonNeumanHierarchyAlternateDescriptionFBExamples}
    &&\frac{d}{dt}g_{1}(t,1)=-\mathcal{N}_{1}(1)g_{1}(t,1),\nonumber\\
    &&\frac{d}{dt}g_{2}(t,1,2)=-\mathcal{N}_{2}(1,2)g_{2}(t,1,2)-\mathcal{N}_{\mathrm{\mathrm{int}}}^{(2)}(1,2)
       \mathcal{S}^{\pm}_{2}g_{1}(t,1)g_{1}(t,2),\nonumber\\
    &&\frac{d}{dt}g_{3}(t,1,2,3)=-\mathcal{N}_{3}(1,2,3)g_{3}(t,1,2,3)-\nonumber\\
    &&\quad-\big(\mathcal{N}_{\mathrm{int}}^{(2)}(1,2)+\mathcal{N}_{\mathrm{int}}^{(2)}(1,3)
       +\mathcal{N}_{\mathrm{int}}^{(3)}(1,2,3)\big)
       \mathcal{S}^{\pm}_{3}g_{1}(t,1)g_{2}(t,2,3)- \nonumber\\
    &&\quad-\big(\mathcal{N}_{\mathrm{int}}^{(2)}(1,2)+\mathcal{N}_{\mathrm{int}}^{(2)}(2,3)+
       \mathcal{N}_{\mathrm{int}}^{(3)}(1,2,3)\big)
       \mathcal{S}^{\pm}_{3}g_{1}(t,2)g_{2}(t,1,3)-\nonumber\\
    &&\quad-\big(\mathcal{N}_{\mathrm{int}}^{(2)}(1,3)+\mathcal{N}_{\mathrm{int}}^{(2)}(2,3)+
       \mathcal{N}_{\mathrm{int}}^{(3)}(1,2,3)\big)
       \mathcal{S}^{\pm}_{3}g_{1}(t,3)g_{2}(t,1,2)-\nonumber\\
    &&\quad-\mathcal{N}_{\mathrm{int}}^{(3)}(1,2,3)\mathcal{S}^{\pm}_{3}(1,2,3)g_{1}(t,1)g_{1}(t,2)g_{1}(t,3).
\end{eqnarray*}
In case of a two-body interaction potential the von Neumann hierarchy (\ref{vonNeumanHierarchyAlternateDescriptionFB}) 
reduces to the following form
\begin{eqnarray}\label{vonNeumannTwoBody}
       &&\frac{d}{dt}g_{s}(t,Y)=-\mathcal{N}_s(Y)g_s(t,Y)+\\
       &&+\sum\limits_{\mathrm{P}:\,Y=X_{1}\bigcup X_2}\,
         \sum\limits_{i_{1}\in X_{1}}\sum\limits_{i_{2}\in X_{2}}
         \big(-\mathcal{N}_{\mathrm{int}}^{(2)}(i_{1},i_{2})\big)
         \mathcal{S}^{\pm}_s g_{|X_{1}|}(t,X_{1})g_{|X_{2}|}(t,X_{2}),\quad s\geq1,\nonumber
\end{eqnarray}
where ${\sum\limits}_{\mathrm{P}:\,Y=X_{1}\bigcup X_2}$ is the sum over all possible partitions $\mathrm{P} $
of the set $Y\equiv(1,\ldots,s)$ into two nonempty mutually disjoint subsets $X_1\subset \ Y$ and $X_2\subset \ Y$.

We note that for classical many-particle systems hierarchy
(\ref{vonNeumannTwoBody}) was introduced in \cite{Gre56} as the zero-order approximation of the nonlinear BBGKY hierarchy for
marginal correlation functions.

In terms of kernels of correlation operators the first two equations of the von Neumann hierarchies
(\ref{vonNeumanHierarchyAlternateDescriptionFB}) and (\ref{vonNeumannTwoBody}) have the form
\begin{eqnarray*}
  &&i\hbar\,\frac{\partial}{\partial t}g_1(t,q_1;q'_1)= -\frac{\hbar^{2}}{2}(\Delta_{q_{1}}-\Delta_{q^{'}_{1}})g_1(t,q_1;q'_1),\\
  &&i\hbar\,\frac{\partial}{\partial t}g_2(t,q_1,q_2;q'_1,q'_2)= \\
  &&=\big(-\frac{\hbar^{2}}{2}\sum_{i=1}^2(\Delta_{q_{i}}-
      \Delta_{q^{'}_{i}})+\big(\Phi^{(2)}(q_1-q_2)-\Phi^{(2)}(q'_1-q'_2)\big)\big)g_2(t,q_1,q_2;q'_1,q'_2)+\\
  &&+\big(\Phi^{(2)}(q_1-q_2)-\Phi^{(2)}(q'_1-q'_2)\big)\frac{1}{2!}\big(g_1(t,q_1;q'_1)g_1(t,q_2;q'_2)\pm g_1(t,q_2;q'_1)g_1(t,q_1;q'_2)\big),
\end{eqnarray*}
where $\Phi^{(2)}$ is a two-body interaction potential (\ref{H_Zag}).

We remark that a particular solution of the steady von Neumann hierarchy is a sequence of the Ursell operators, 
for example, its first two elements have the form
\begin{eqnarray*}
  &&g_1(t,1)= e^{-\beta K(1)},\\
  &&g_2(t,1,2)=\mathcal{S}^{\pm}_2 e^{-\beta \sum\limits_{i=1}^{2}K(i)}(e^{-\beta\Phi^{(2)}(1,2)}-I),
\end{eqnarray*}
where $\beta$ is a parameter inversely proportional to temperature, the operators $K(i)$ and $\Phi^{(2)}$ are defined by (\ref{H_Zag}).
We use the name Ursell operators to indicate their intimate structure
with the expansions of classical equilibrium states developed by H.D. Ursell \cite{Urs}.

\subsection{A formal solution of the von Neumann hierarchy}
To construct a solution of the Cauchy problem (\ref{vonNeumanHierarchyAlternateDescriptionFB}),(\ref{vonNeumanHierarchyAlternateDescriptionInitialValue})
we first consider its structure for one physically motivated example of initial data,
namely initial data satisfying a chaos property (statistically independent particles) \cite{CGP97}.
A chaos property means that there are no correlations in a system at the initial instant.
In this case the sequence of initial correlation operators is the one-component sequence
\begin{eqnarray}\label{posl_g(0)}
     && g(0)=(0,g_{1}(0,1),0,\ldots).
\end{eqnarray}
In fact, in terms of the sequence of the density operators this condition means that
$D(0)=(1,D_{1}(0,1), \mathcal{S}^{\pm}_{2}D_{1}(0,1)D_{1}(0,2),\ldots,\mathcal{S}^{\pm}_{n}\prod^n_{i=1} D_{1}(0,i),\ldots)$.

On basis of representation (\ref{gfromDFB}) of the correlation operators
in terms of the density operators and formula (\ref{rozv_fon-N}), in view of the fact
that the operators $\mathcal{S}^{\pm}_{2}$ commute with groups of operators (\ref{groupG}), we have
\begin{eqnarray*}
   &&g_{s}(t,Y)=\sum\limits_{\mbox{\scriptsize $\begin{array}{c}\mathrm{P}:Y=\bigcup_{i}X_{i}\end{array}$}}\hskip -1mm
       (-1)^{|\mathrm{P}|-1}(|\mathrm{P}|-1)!\,\prod_{X_i\subset \mathrm{P}}\mathcal{G}_{|X_{i}|}(-t,X_{i})
       \,{\mathcal{S}^{\pm}_{s}}\prod_{i=1}^{s}D_{1}(0,i), \quad s\geq1.
\end{eqnarray*}
Taking into account the equality: $D_{1}(0,i)=g_{1}(0,i), \, 1\leq i\leq s$,  we derive
the formula of a solution of the Cauchy problem (\ref{vonNeumanHierarchyAlternateDescriptionFB}),(\ref{vonNeumanHierarchyAlternateDescriptionInitialValue})
for initial data (\ref{posl_g(0)})
\begin{eqnarray}\label{rozvChaosN}
   && g_{s}(t,Y)=\mathfrak{A}_{s}(t,Y)\,
      \mathcal{S}^{\pm}_s \,\prod_{i=1}^{s}g_{1}(0,i), \quad s\geq1,
\end{eqnarray}
where $\mathfrak{A}_{s}(t,Y)$ is the $s$-order cumulant\footnote{Some properties of cumulants are considered in section 3.1.}
of the groups of operators (\ref{groupG}) of the von Neumann equations \cite{GerS,GerShJ}
\begin{eqnarray} \label{cumulants}
    &&\mathfrak{A}_{s}(t,Y)\doteq\sum\limits_{\mathrm{P}:\,Y =
       \bigcup_i X_i}(-1)^{|\mathrm{P}|-1}({|\mathrm{P}|-1})!
       \prod\limits_{X_i\subset\mathrm{P}}\mathcal{G}_{|X_{i}|}(-t,X_{i}),
\end{eqnarray}
and we use accepted above notations. It should be emphasized that
the structure of expansions (\ref{cumulants}) of the evolution operators $\mathfrak{A}_{s}(t,Y)$ means that they are determed
by the cluster expansions of the groups of operators (\ref{groupG}) of the von Neumann equation (\ref{vonNeumannEqn})-(\ref{F-N12}).

Thus, the cumulant nature of correlation operators induces the cumulant structure of a one-parametric mapping
generated by solution (\ref{rozvChaosN}) and as we shall see in section 4, the same statement is true for
the cumulant structure of the solution expansion of the BBGKY hierarchy.

From (\ref{rozvChaosN}) it is clear that in case of absence of correlations in a system at initial instant
the correlations generated by the dynamics of a system are completely governed by the cumulants (\ref{cumulants})
of evolution operators (\ref{groupG}).

Hereafter we use the following notations: $Y_\mathrm{P}\equiv(\{X_1\},\ldots,\{X_{|\mathrm{P}|}\})$ is a set,
elements of which are $|\mathrm{P}|$ mutually disjoint subsets $X_i\subset Y\equiv(1,\ldots,s)$ of the partition $\mathrm{P}:Y=\cup_{i=1}^{|\mathrm{P}|}X_i$, i.e. $|Y_\mathrm{P}|=|\mathrm{P}|$. In view of these notations we state that $\{Y\}$
is the set consisting of one element $Y=(1,\ldots,s)$ of the partition $\mathrm{P}$ $(|\mathrm{P}|=1)$ and $|\{Y\}|=1$.

We define the declasterization mapping $\theta: Y_P\rightarrow Y$, by the following formula
\begin{eqnarray}\label{Theta}
   &&\theta(Y_P)=Y.
\end{eqnarray}
For example, let $X\equiv(1,\ldots,s+n)$, then for the set $(\{Y\},X\setminus Y)$ it holds
\begin{eqnarray*}
   &&\theta(\{Y\},X\setminus Y) = X .
\end{eqnarray*}

Now we construct a solution of hierarchy
(\ref{vonNeumanHierarchyAlternateDescriptionFB}) for arbitrary initial data.
Using solution (\ref{rozv_fon-N}) of the von Neumann equation (\ref{vonNeumannEqn})
we rewrite formula (\ref{gfromDFB}) as follows
\begin{eqnarray*}
   &&g_{s}(t,Y)=\sum\limits_{\mbox{\scriptsize $\begin{array}{c}\mathrm{P}:Y=\bigcup_{i}X_{i}\end{array}$}}
       (-1)^{|\mathrm{P}|-1}(|\mathrm{P}|-1)!\,{\mathcal{S}^{\pm}_{s}}
       \prod_{X_i\subset \mathrm{P}}\mathcal{G}_{|X_i|}(-t,X_i)D_{|X_i|}(0,X_i).
\end{eqnarray*}
In view of expansion (\ref{D_(g)N}) at initial instant we have
\begin{eqnarray*}
   &&g_{s}(t,Y)=\hskip -2mm\sum\limits_{\mbox{\scriptsize $\begin{array}{c}\mathrm{P}:Y=\bigcup_{i}X_{i}\end{array}$}}\hskip -1mm
       (-1)^{|\mathrm{P}|-1}(|\mathrm{P}|-1)!{\mathcal{S}^{\pm}_{s}}
       \prod_{X_i\subset \mathrm{P}}\mathcal{G}_{|X_i|}(-t,X_i)\sum\limits_{\mathrm{P'}:X_i=
       \bigcup_{j_i} Z_{j_i}}\,\prod_{Z_{j_i}\subset \mathrm{P'}}g_{|Z_{j_i}|}(0,Z_{j_i}).
\end{eqnarray*}
Changing the summation indexes $\mathrm{P}$ and $\mathrm{P'}$ in this expression gives us
\begin{eqnarray*}
    &&g_{s}(t,Y)=\sum\limits_{\mbox{\scriptsize $\begin{array}{c}\mathrm{P}:Y=\bigcup_{i}X_{i}\end{array}$}}
       \sum\limits_{\mathrm{P'}:(\{X_1\},\ldots,\{X_{|\mathrm{P}|}\})=\bigcup_k Z_k}
       (-1)^{|\mathrm{P'}|-1}(|\mathrm{P'}|-1)! \times\\
    &&\times\prod_{Z_k\subset \mathrm{P'}}\mathcal{G}_{|\theta(Z_k)|}(-t,\theta(Z_k)){\mathcal{S}^{\pm}_{s}}
       \prod_{X_i\subset \mathrm{P}}g_{|X_i|}(0,X_i).
\end{eqnarray*}
Finally we represent the solution of Cauchy problem  (\ref{vonNeumanHierarchyAlternateDescriptionFB}),(\ref{vonNeumanHierarchyAlternateDescriptionInitialValue})
for arbitrary initial data in the following form
\begin{eqnarray}\label{rozvNF-N_F}
    &&g_{s}(t,Y)=\sum\limits_{\mathrm{P}:\,Y=\bigcup_iX_i}
        \mathfrak{A}_{|\mathrm{P}|}(t,\{X_1\},\ldots,\{X_{|\mathrm{P}|}\})
        \mathcal{S}^{\pm}_{s}\prod_{X_i\subset \mathrm{P}}g_{|X_i|}(0,X_i),\quad s\geq1.
\end{eqnarray}
Here $\mathfrak{A}_{|\mathrm{P}|}(t)$ is the $|\mathrm{P}|$-order cumulant of groups of operators (\ref{groupG})
defined by the formula
\begin{eqnarray} \label{cumulantP}
    &&\mathfrak{A}_{|\mathrm{P}|}(t,\{X_1\},\ldots,\{X_{|\mathrm{P}|}\})\doteq\\
    &&\doteq\sum\limits_{\mathrm{P}^{'}:\,(\{X_1\},\ldots,\{X_{|\mathrm{P}|}\})=
    \bigcup_k Z_k}(-1)^{|\mathrm{P}^{'}|-1}({|\mathrm{P}^{'}|-1})!
       \prod\limits_{Z_k\subset\mathrm{P}^{'}}\mathcal{G}_{|\theta(Z_{k})|}(-t,\theta(Z_{k})),\nonumber
\end{eqnarray}
where $\sum_{\mathrm{P}^{'}:\,(\{X_1\},\ldots,\{X_{|\mathrm{P}|}\})=\bigcup_k Z_k}$
is the sum over all possible partitions $\mathrm{P}^{'}$
of the set $(\{X_1\},\ldots,$ $\{X_{|\mathrm{P}|}\})$ into $|\mathrm{P}^{'}|$ nonempty mutually disjoint subsets
$Z_k\subset (\{X_1\},\ldots,$ $\{X_{|\mathrm{P}|}\})$.

We remark that for $|\mathrm{P}|\geq 2$, the  $|\mathrm{P}|$-order cumulants $\mathfrak{A}_{|\mathrm{P}|}(t)$
of groups of operators (\ref{groupG}) of the von Neumann equations have similar structure
in contrast to the first-order cumulant. Indeed, for $|\mathrm{P}|=1$ we have
\begin{eqnarray*}
  &&\mathfrak{A}_{1}(t,\{1,\ldots,s\})=\mathcal{G}_{n}(-t,1,\ldots,s),
\end{eqnarray*}
and, for instance, in case of $|\mathrm{P}|=s$, according to (\ref{cumulantP}),
the $s$-order cumulant has structure (\ref{cumulants}).

The simplest examples of correlation operators (\ref{rozvNF-N_F}) are given by the expressions
\begin{eqnarray*}
    &&g_{1}(t,1)=\mathfrak{A}_{1}(t,1)g_{1}(0,1),\\
    &&g_{2}(t,1,2)=\mathfrak{A}_{1}(t,\{1,2\})g_{2}(0,1,2)+\mathfrak{A}_{2}(t,1,2)\mathcal{S}^{\pm}_2g_{1}(0,1)g_{1}(0,2),\\
    &&g_{3}(t,1,2,3)= \mathfrak{A}_{1}(t,\{1,2,3\})g_{3}(0,1,2,3)+ \mathfrak{A}_{2}(t,\{2,3\},1)\mathcal{S}^{\pm}_3g_{1}(0,1)g_{2}(0,2,3)+\\
    &&\qquad+\mathfrak{A}_{2}(t,\{1,3\},2)\mathcal{S}^{\pm}_3g_{1}(0,2)g_{2}(0,1,3)+
        \mathfrak{A}_{2}(t,\{1,2\},3)\mathcal{S}^{\pm}_3g_{1}(0,3)g_{2}(0,1,2)+\\
    &&\qquad+\mathfrak{A}_{3}(t,1,2,3)\mathcal{S}^{\pm}_3g_{1}(0,1)g_{1}(0,2)g_{1}(0,3).
\end{eqnarray*}

The validity of solution expansion (\ref{rozvNF-N_F}) can be verified by
straightforward differentiation by time variable (see section 3) and also in the following way.
Taking into account the fact that the von Neumann hierarchy (\ref{vonNeumanHierarchyAlternateDescriptionFB})
is the evolution recurrence equations set, we can construct a solution of initial-value problem  (\ref{vonNeumanHierarchyAlternateDescriptionFB}),(\ref{vonNeumanHierarchyAlternateDescriptionInitialValue})
by integrating each equation of the hierarchy as the inhomogeneous von Neumann equation.
For example, as a result of integrating the first two equations of hierarchy (\ref{vonNeumanHierarchyAlternateDescriptionFB})
and using formula (\ref{groupG}), we obtain the following equalities
\begin{eqnarray*}
    &&g_{1}(t,1)=\mathcal{G}_{1}(-t,1)g_{1}(0,1),\\
    &&g_{2}(t,1,2)=\mathcal{G}_{2}(-t,1,2)g_{2}(0,1,2)+\\
    &&+\int\limits_{0}^{t}dt_{1}\mathcal{G}_{2}(-t+t_{1},1,2)\big(-\mathcal{N}^{(2)}_{\mathrm{int}}(1,2)\big)
       \mathcal{S}^{\pm}_{2}\mathcal{G}_{1}(-t_{1},1)\mathcal{G}_{1}(-t_{1},2)g_{1}(0,1)g_{1}(0,2).
\end{eqnarray*}
In consequence of the fact that the operators $\mathcal{S}^{\pm}_{2}$ commute with groups
$\mathcal{G}_{1}(-t_{1},1), \mathcal{G}_{1}(-t_{1},2)$,
the second equality takes the form
\begin{eqnarray*}
     &&g_{2}(t,1,2)=\mathcal{G}_{2}(-t,1,2)g_{2}(0,1,2)+\\
     &&+\int\limits_{0}^{t}dt_{1}\mathcal{G}_{2}(-t+t_{1},1,2)\big(-\mathcal{N}^{(2)}_{\mathrm{int}}(1,2)\big)
        \mathcal{G}_{1}(-t_{1},1)\mathcal{G}_{1}(-t_{1},2)\mathcal{S}^{\pm}_{2}g_{1}(0,1)g_{1}(0,2).
\end{eqnarray*}
Then for the second term on the right-hand side of this equation an analog of the Duhamel equation holds
\begin{eqnarray}\label{iter2kum}
    &&\int\limits_{0}^{t}dt_{1}\mathcal{G}_{2}(-t+t_{1},1,2)
         \big(-\mathcal{N}^{(2)}_{\mathrm{int}}(1,2)\big)\mathcal{G}_{1}(-t_{1},1)\mathcal{G}_{1}(-t_{1},2)=\\
    &&=-\mathcal{G}_{2}(-t,1,2)\int\limits_{0}^{t}dt_{1}\frac{d}{dt_{1}}\big(\mathcal{G}_{2}(t_{1},1,2)
         \mathcal{G}_{1}(-t_{1},1)\mathcal{G}_{1}(-t_{1},2)\big)=\nonumber\\
    &&=\mathcal{G}_{2}(-t,1,2)-\mathcal{G}_{1}(-t,1)\mathcal{G}_{1}(-t,2)=\mathfrak{A}_{2}(t,1,2),\nonumber
\end{eqnarray}
where $\mathfrak{A}_{2}(t)$ is the second-order cumulant of groups of operators (\ref{groupG}) defined by formula (\ref{cumulants}).
For $s>2$ the solution of the Cauchy problem (\ref{vonNeumanHierarchyAlternateDescriptionFB}),(\ref{vonNeumanHierarchyAlternateDescriptionInitialValue}),
which is constructed by iterations, is represented by expansions (\ref{rozvNF-N_F})
as a consequence of transformations similar to an analog of the Duhamel equation (\ref{iter2kum}).

We note, that in case of initial data (\ref{posl_g(0)}) solution (\ref{rozvChaosN}) of the Cauchy problem (\ref{vonNeumanHierarchyAlternateDescriptionFB}),(\ref{vonNeumanHierarchyAlternateDescriptionInitialValue})
of the von Neumann hierarchy may be rewritten in another representation. For $n=1$, we have
\begin{eqnarray*}
     && g_{1}(t,1)=\mathfrak{A}_{1}(t,1)g_{1}(0,1).
 \end{eqnarray*}
Then, within the context of the definition of the first-order cumulant,
$\mathfrak{A}_{1}(t)$, and the dual group of operators
$\mathfrak{A}_{1}(-t)$, we express the correlation operators
$g_{s}(t), \,s\geq 2$, in terms of the one-particle correlation
operator $g_{1}(t)$ using formula
(\ref{rozvChaosN}). Hence for $s\geq2$ formula (\ref{rozvChaosN}) is represented in the form
\begin{eqnarray*}
       &&g_{s}(t,Y)=\widehat{\mathfrak{A}}_{s}(t,Y)\,
          \mathcal{S}^{\pm}_s \,\prod_{i=1}^{s}\,g_{1}(t,i),\quad s\geq 2,
\end{eqnarray*}
where $\widehat{\mathfrak{A}}_{s}(t,Y)$ is the $s$-order
cumulant (\ref{cumulants}) of the scattering operators
\begin{eqnarray*}
     &&\widehat{\mathcal{G}}_{t}(Y)\doteq
        \mathcal{G}_{s}(-t,Y)\prod_{k=1}^{s}\mathcal{G}_{1}(t,k),\quad s\geq1.
\end{eqnarray*}
The generator of the scattering operator $\widehat{\mathcal{G}}_{t}(Y)$ is determined by the operator
\begin{eqnarray*}
  &&\frac{d}{dt}\widehat{\mathcal{G}}_{t}(Y)|_{t=0}= -\sum\limits_{k=2}^{s}\,\,\sum\limits_{i_{1}<\ldots<i_{k}=1}^{s}
     \,\mathcal{N}_{\mathrm{int}}^{(k)}(i_{1},\ldots,i_{k}),
\end{eqnarray*}
where the operator $\mathcal{N}_{\mathrm{int}}^{(k)}$ acts in $\mathfrak{L}^{1}_{0}(\mathcal{H}^{\pm}_{s})\subset
\mathfrak{L}^{1}(\mathcal{H}^{\pm}_{s})$ according to formula (\ref{oper Nint2}).

\subsection{Correlation operators of particle clusters}
Further we consider a more general notion, namely, the correlation operators of particle clusters
that describe the correlations between clusters of particles.
For a system consisting of particles and of one particle cluster $\{Y\}\equiv\{1,\ldots,s\}$ we introduce the sequence of operators
$g^{(s)}(t)=(g_{1+0}(t,\{Y\}),\ldots,g_{1+n}(t,\{Y\},X\setminus Y),\ldots) \in \mathfrak{L}^{1}(\oplus_{n=0}^{\infty}\mathcal{H}_{s+n}^{\pm})$
that are defined by the following generalized cluster expansions of the density operators
$D(t)=(1,D_{1}(t,1),\ldots,$ $D_{s+n}(t,1,\ldots,s+n),\ldots)\in \mathfrak{L}^{1}(\mathcal{F}^{\pm}_\mathcal{H})$
\begin{eqnarray}\label{gce}
  &&D_{s+n}(t,X)= g_{1+n}(t,\{Y\},X\setminus Y)+\\
  &&+\sum\limits_{\mbox{\scriptsize $\begin{array}{c}\mathrm{P}:(\{Y\},X\setminus Y)=\bigcup_{i}X_{i},\\|\mathrm{P}|>1\end{array}$}}{\mathcal{S}^{\pm}_{s+n}}
        \prod_{X_i\subset \mathrm{P}}g_{|X_i|}(t,X_i), \qquad n\geq0,\nonumber
\end{eqnarray}
where $X\equiv(1,\ldots,s+n)$, $X \setminus Y=(s+1,\ldots,s+n)$ and ${\sum\limits}_{
\mathrm{P}:\,(\{Y\},\,X\setminus Y)=\,\bigcup_{i} X_{i},\,|\mathrm{P}|>1}$
is the sum over all possible partitions $\mathrm{P} $ of the set $(\{Y\},X\setminus Y)$ into $|\mathrm{P}|>1$
nonempty mutually disjoint subsets $X_i\subset (\{Y\},X\setminus Y)$.

Relations between correlation operators of particle clusters $g^{(s)}(t)=(g_{1+0}(t,\{Y\}),\ldots,g_{1+n}(t,$ $\{Y\},X\setminus Y),\ldots) \in \mathfrak{L}^{1}(\oplus_{n=0}^{\infty}\mathcal{H}_{s+n}^{\pm})$ and correlation operators of particles (\ref{gfromDFB})
is given by the equality
\begin{eqnarray}\label{gCluster}
  &&g_{1+n}(t,\{Y\},X\setminus Y)=\\
  &&=\sum\limits_{\mathrm{P}:(\{Y\},\,X\setminus Y)=\bigcup_i X_i}
      (-1)^{|\mathrm{P}|-1}(|\mathrm{P}| -1)!\, \mathcal{S}^{\pm}_{s+n}\prod_{X_i\subset \mathrm{P}}\,\,\sum\limits_{\mathrm{P'}:\,\theta(X_{i})=\bigcup_{j_i} Z_{j_i}}\hskip 1mm
      \prod_{Z_{j_i}\subset \mathrm{P'}}g_{|Z_{j_i}|}(t,Z_{j_i}).\nonumber
\end{eqnarray}
In particular case $n=0$, i.e. the correlation operator of a cluster of $|Y|$ particles, these relations take the form
\begin{eqnarray*}\label{gCluster0}
  &&g_{1+0}(t,\{Y\})=\mathcal{S}^{\pm}_{s}\sum\limits_{\mathrm{P}:\,Y=\bigcup_{i} X_{i}}\hskip 1mm
      \prod_{X_{i}\subset \mathrm{P}}g_{|X_{i}|}(t,X_{i}).\nonumber
\end{eqnarray*}

From relations (\ref{gCluster}) using hierarchy (\ref{vonNeumanHierarchyAlternateDescriptionFB})
we derive the von Neumann hierarchy for the correlation operators of particle clusters 
\begin{eqnarray}\label{gClusterhierarchy}
    &&\frac{d}{dt}g_{1+n}(t,\{Y\},X\setminus Y)
       =-\mathcal{N}_{s+n}(X)g_{1+n}(t,\{Y\},X\setminus Y)+\\
    &&+\sum\limits_{\mbox{\scriptsize $\begin{array}{c}{\mathrm{P}}:(\{Y\},\,X\setminus Y)={\bigcup\limits}_i X_i,\\|\mathrm{P}|>1\end{array}$}}
       \sum\limits_{\mbox{\scriptsize $\begin{array}{c} Z_1 \subset \theta(X_1),\\Z_{1}\neq\emptyset\end{array}$}}
       \ldots \sum\limits_{\mbox{\scriptsize $\begin{array}{c}Z_{|\mathrm{P}|}\subset\theta(X_{|\mathrm{P}|}),\\Z_{|\mathrm{P}|}\neq\emptyset\end{array}$}}
       \big(-\mathcal{N}_{\mathrm{int}}^{(\sum\limits_{i=1}^{|\mathrm{P}|}|Z_{i}|)}(Z_1,\ldots, \nonumber\\
    &&Z_{|\mathrm{P}|})\big)
       \mathcal{S}^{\pm}_{s+n}\prod\limits_{X_i \subset \mathrm{P}}g_{|X_i|}(t,X_{i}).\nonumber
\end{eqnarray}

The formal solution of the von Neumann hierarchy for correlation operators 
of clusters of particles is given by the following expansion
\begin{eqnarray}\label{rozvNF-N_F_clusters}
    &&g_{1+n}(t,\{Y\},X\setminus Y)=\\
    &&=\sum\limits_{\mathrm{P}:\,(\{Y\},\,X\setminus Y)=\bigcup_iX_i}
        \mathfrak{A}_{|\mathrm{P}|}\big(t,\{\theta(X_1)\},\ldots,\{\theta(X_{|\mathrm{P}|})\}\big)
        \mathcal{S}^{\pm}_{s+n}\prod_{X_i\subset \mathrm{P}}g_{|X_i|}(0,X_i),\quad n\geq 0,\nonumber
\end{eqnarray}
where $\mathfrak{A}_{|\mathrm{P}|}(t)$ is the $|\mathrm{P}|$-order cumulant defined by (\ref{cumulantP}).

For the sequence of correlation operators of clusters of particles $g^{(s)}(t)=(g_{1+0}(t,\{Y\}),g_{1+1}(t,$
$\{Y\},s+1),\ldots,g_{1+n}(\{Y\},s+1,\ldots,s+n),\ldots)$ the initial data satisfying a chaos property, i.e.
in case of absence of correlations between particles and between particle cluster and particles, is the sequence
\begin{eqnarray}\label{gChaos}
  &&g^{(s)}(0)=(g_{1+0}(0,\{Y\}),0,\ldots).
\end{eqnarray}
The corresponding solution of the initial-value problem of the von Neumann hierarchy is given by the expansion
\begin{eqnarray}\label{gtUg0}
   &&g_{1+n}(t,\{Y\},X\setminus Y)=\mathfrak{A}_{1+n}(t,\{Y\},\,X\setminus Y)\,
       \mathcal{S}_{s+n}^{\pm}\prod\limits_{i\in (\{Y\},\,X\setminus Y)}g_{1}(0,i),
\end{eqnarray}
where $\mathfrak{A}_{1+n}(t)$ is the $(n+1)$-order cumulant defined by formula (\ref{cumulantP}) and has the form
\begin{eqnarray}\label{cumcp}
    &&\mathfrak{A}_{1+n}(t,\{Y\},\,X\setminus Y)=\sum\limits_{\mathrm{P}:\,(\{Y\},\,X\setminus Y)=
    \bigcup_i X_i}(-1)^{|\mathrm{P}|-1}({|\mathrm{P}|-1})!
       \prod\limits_{X_i\subset\mathrm{P}}\mathcal{G}_{|\theta(X_i)|}(-t,\theta(X_i)).
\end{eqnarray}
Thus, correlations created
during the evolution of a system are described by formula (\ref{gtUg0}) and defined by the corresponding
order cumulant (\ref{cumcp}).

\section{The Cauchy problem of the von Neumann hierarchy in the spaces of sequences of trace-class operators}

\subsection{Some properties of a solution}
The intrinsic properties of constructed solution (\ref{rozvNF-N_F}) are generated by the properties
of cumulants (\ref{cumulantP}) of groups of operators of the von Neumann equations.

At initial time $t=0$ solution (\ref{rozvNF-N_F}) satisfies
initial condition (\ref{vonNeumanHierarchyAlternateDescriptionInitialValue}). Indeed,
according to definition (\ref{groupG}), i.e. that $\mathcal{G}_{n}(0)=I$ is a unit operator,
and in view of the identity
\begin{eqnarray}\label{Stirl}
    && \sum\limits_{\mathrm{P}:\,(1,\ldots,n)=\bigcup_iX_i}
      (-1)^{| \mathrm{P}|-1}(|\mathrm{P}|-1)!=
      \sum\limits_{k=1}^{n}(-1)^{k-1}\mathrm{s}(n,k)(k-1)!=\delta_{n,1},
\end{eqnarray}
where $\mathrm{s}(n,k)$ are the Stirling numbers of the second kind and $\delta_{n,1}$ is a Kronecker symbol,
for $|\mathrm{P}|\geq 2$ we have
\begin{eqnarray*}
  &&\mathfrak{A}_{|\mathrm{P}|}(0,\{X_1\},\ldots,\{X_{|\mathrm{P}|}\})=
      \sum\limits_{\mathrm{P}^{'}:\,(\{X_1\},\ldots,\{X_{|\mathrm{P}|}\})=\bigcup_k Z_k}
      (-1)^{| \mathrm{P}^{'}|-1}(|\mathrm{P}^{'}|-1)!I=0.
\end{eqnarray*}

Let us consider some properties of cumulants (\ref{cumulants}) of groups of operators (\ref{groupG})
introduced in previous section.

Since for a system of non-interacting particles it holds
\begin{eqnarray*}
    &&\mathcal{G}_{n}(-t,1,\ldots,n)=\prod_{i=1}^{n}\mathcal{G}_{1}(-t,i),
\end{eqnarray*}
then taking into consideration identity (\ref{Stirl}), for $n\geq 2$ we obtain
\begin{eqnarray*}
    &&\mathfrak{A}_{n}(t,1,\ldots,n)=
       \sum\limits_{\mathrm{P}:\,(1,\ldots,n)=\bigcup_i X_i}
       (-1)^{| \mathrm{P}|-1}(|\mathrm{P}|-1)!\prod\limits_{X_{i}\subset \mathrm{P}}
       \prod_{l_i=1}^{|X_{i}|}\mathcal{G}_{1}(-t,l_i)=\\
    &&=\sum\limits_{k=1}^{n}(-1)^{k-1}\mathrm{s}(n,k)(k-1)!\prod_{i=1}^{n}\mathcal{G}_{1}(-t,i)=0.
\end{eqnarray*}

The infinitesimal generator of the first-order cumulant is defined by the following
limit in the sense of the norm convergence in the spaces $\mathfrak{L}^{1}(\mathcal{H}^{\pm}_{n})$
\begin{eqnarray*}
    &&\lim\limits_{t\rightarrow 0}\frac{1}{t}(\mathfrak{A}_{1}(t,\{1,\ldots,n\})-I) f_{n}
       =(-\mathcal{N}_{n}f_{n})(1,\ldots,n),
\end{eqnarray*}
where the operator $(-\mathcal{N}_{n})$ is defined by formula (\ref{komyt}) for $f_n\in\mathfrak{L}^{1}_0(\mathcal{H}^{\pm}_{n})\subset\mathfrak{L}^{1}(\mathcal{H}^{\pm}_{n})$.

In general case, i.e. $n\geq 2$, the generator of the $n$-order cumulant is given by the operator
\begin{eqnarray}\label{diffCumulant}
    &&\lim\limits_{t\rightarrow 0}\frac{1}{t}\mathfrak{A}_{n}(t,1,\ldots,n) f_{n}=
       (-\mathcal{N}^{(n)}_{\mathrm{int}}f_{n})(1,\ldots,n),
\end{eqnarray}
defined by formula (\ref{oper Nint2})
on $\mathfrak{L}^{1}_0(\mathcal{H}^{\pm}_{n})\subset\mathfrak{L}^{1}(\mathcal{H}^{\pm}_{n})$ and the limit exists
in the sense of the norm convergence in the spaces $\mathfrak{L}^{1}(\mathcal{H}^{\pm}_{n})$.
To prove the validity of statement (\ref{diffCumulant}) we use definition (\ref{cumulants}) and
equality (\ref{infOper}) and represent the limit in left hand side of (\ref{diffCumulant}) in the explicit form
\begin{eqnarray}\label{Nint}
   &&\lim\limits_{t\rightarrow 0}\frac{1}{t}\mathfrak{A}_{n}(t) f_{n}
       =\lim\limits_{t\rightarrow 0}\frac{1}{t}\sum\limits_{\mathrm{P}:\,(1,\ldots,n)=\bigcup_k Z_k}(-1)^{|\mathrm{P}|-1}
       ({|\mathrm{P}|-1})! \prod\limits_{Z_k\subset\mathrm{P}}\mathcal{G}_{|Z_{k}|}(-t,Z_{k})f_{n}=\\
   &&=\sum\limits_{\mathrm{P}:\,(1,\ldots,n)=\bigcup_k Z_k}(-1)^{|\mathrm{P}|-1}(|\mathrm{P}|-1)! \sum\limits_{Z_k\subset
       \mathrm{P}}(-\mathcal{N}_{|Z_{k}|}(Z_k))f_{n}=\nonumber\\
   &&=\sum\limits_{\mathrm{P}:\,(1,\ldots,n)=\bigcup_k Z_k}(-1)^{|\mathrm{P}|-1}(|\mathrm{P}|-1)!\,
       \sum\limits_{i=1}^{n}(-\mathcal{N}_{1}(i))f_{n}+\nonumber\\
   &&+ \sum\limits_{\mathrm{P}:\,(1,\ldots,n)=\bigcup_k Z_k}(-1)^{|\mathrm{P}|-1}(|\mathrm{P}|-1)! \sum\limits_{Z_k\subset
       \mathrm{P}}\,\sum\limits^{|Z_{k}|}_{r=2}\, \sum\limits^{|Z_{k}|}_{j_{1}<\ldots<j_{r}=r}
       (-\mathcal{N}^{(r)}_{\mathrm{int}}(j_{1},\ldots,j_{n})) f_{n},\nonumber
\end{eqnarray}
where $\mathcal{N}_{1}(i)f_{n}\equiv \frac{i}{\hbar}(K(i)f_{n}-f_{n}K(i))$.
Taking into account identity (\ref{Stirl}) and the fact that $n\geq 2$,
for the first term of formula (\ref{Nint}) it holds
\begin{eqnarray*}
    &&\sum\limits_{\mathrm{P}:\,(1,\ldots,n)=\bigcup_k Z_k}\hskip-2mm
       (-1)^{|\mathrm{P}|-1}(|\mathrm{P}|-1)!\sum\limits_{i=1}^{n}(-\mathcal{N}_{1}(i))=0.
\end{eqnarray*}
Applying identity (\ref{Stirl}) to the result of rearrangements in the second term of (\ref{Nint}), we finally obtain
\begin{eqnarray*}
    &&\sum\limits_{\mathrm{P}:\,(1,\ldots,n)=\bigcup_k Z_k}(-1)^{|\mathrm{P}|-1}(|\mathrm{P}|-1)!\sum\limits_{Z_k\subset
       \mathrm{P}}\sum\limits^{|Z_{k}|}_{r=2} \sum\limits^{|Z_{k}|}_{j_{1}<\ldots<j_{r}=r} (-\mathcal{N}^{(r)}_{\mathrm{int}}(j_{1},\ldots,j_{n}))=
\end{eqnarray*}
\begin{eqnarray*}
    &&=\sum\limits_{Z\subset Y}(-\mathcal{N}_{\mathrm{int}}^{(|Z|)}(Z))\sum\limits_{k=1}^{n-|Z|+1}(-1)^{|\mathrm{P}|-1}(|\mathrm{P}|-1)!
       \,\mathrm{s}(n-|Z|+1,k)=\nonumber\\
    &&=\sum\limits_{Z\subset Y}(-\mathcal{N}_{\mathrm{int}}^{(|Z|)}(Z))\,\delta_{n,|Z|}=-\mathcal{N}^{(n)}_{\mathrm{int}}(1,\ldots,n),
\end{eqnarray*}
We emphasize that the operator $(-\mathcal{N}^{(n)}_{\mathrm{int}})$ is defined by a $n$-body interaction potential (\ref{H_Zag}).

Similarly to the proof of statement (\ref{diffCumulant}) for the $|\mathrm{P}|$-order, $|\mathrm{P}|\geq 2$, cumulant (\ref{cumulantP})
of particle clusters for $f_n\in\mathfrak{L}^{1}_0(\mathcal{H}^{\pm}_{n})\subset\mathfrak{L}^{1}(\mathcal{H}^{\pm}_{n})$ we obtain
in the sense of the norm convergence in the spaces $\mathfrak{L}^{1}(\mathcal{H}^{\pm}_{n})$
\begin{eqnarray}\label{oper Nint1}
    &&\lim\limits_{t\rightarrow 0}\frac{1}{t}\mathfrak{A}_{|\mathrm{P}|}(t,\{X_1\},\ldots,\{X_{|\mathrm{P}|}\}) f_{n}=\\
    &&=\sum\limits_{\mathrm{P}^{'}:\,(\{X_1\},\ldots,\{X_{|\mathrm{P}|}\})=\bigcup_k Z_k}
        (-1)^{|\mathrm{P}^{'}|-1}(|\mathrm{P}^{'}| -1)!\sum\limits_{Z_k\subset \mathrm{P}^{'}}
        (-\mathcal{N}_{|\theta(Z_{k})|}(\theta(Z_k)))f_{n}=\nonumber\\
    &&=\sum\limits_{\mbox{\scriptsize
        $\begin{array}{c}{Z_{1}\subset X_{1}},\\Z_{1}\neq\emptyset\end{array}$}}\ldots
        \sum\limits_{\mbox{\scriptsize
        $\begin{array}{c}{Z_{|\mathrm{P}|}\subset X_{|\mathrm{P}|}},\\Z_{|\mathrm{P}|}\neq\emptyset\end{array}$}}
        \big(-\mathcal{N}_{\mathrm{int}}^{(\sum\limits_{r=1}^{|\mathrm{P}|}|Z_{{r}}|)}
        (Z_{{1}},\ldots,Z_{{|\mathrm{P}|}})\big)f_{n}.\nonumber
\end{eqnarray}

\subsection{The group of nonlinear operators generated by the von Neumann hierarchy}
In the spaces $\mathfrak{L}^{1}(\mathcal{F}^{\pm}_\mathcal{H})$ constructed solution (\ref{rozvNF-N_F})
generates a group of nonlinear operators.
The properties of this group are described by the following theorem.
\begin{theorem}
For  $f\in\mathfrak{L}^{1}(\mathcal{F}^{\pm}_\mathcal{H})$ the mapping
\begin{eqnarray}\label{groupKum}
    &&t\rightarrow \big(\mathfrak{A}_{t}(f)\big)_{n}\doteq \sum\limits_{\mathrm{P}:\,(1,\ldots,n)=\bigcup_iX_i}
        \mathfrak{A}_{|\mathrm{P}|}(t,\{X_1\},\ldots,\{X_{|\mathrm{P}|}\})
        \mathcal{S}^{\pm}_{n}\prod_{X_i\subset \mathrm{P}}f_{|X_i|}(X_i)
\end{eqnarray}
is a group of nonlinear operators of class $C_{0}$. In the subspaces $\mathfrak{L}^{1}_{0}(\mathcal{H}^{\pm}_{n})\subset
\mathfrak{L}^{1}(\mathcal{H}^{\pm}_{n})$, $n\geq1$, the infinitesimal generator
$\mathfrak{N}(\cdot)$ of group (\ref{groupKum}) is defined by the operator
\begin{eqnarray}\label{Nnl}
    &&\big(\mathfrak{N}(f)\big)_{n}(1,\ldots,n)\doteq
        -\mathcal{N}_{n}f_{n}+\\
    &&+\sum\limits_{\mbox{\scriptsize $\begin{array}{c}\mathrm{P}:(1,\ldots,n)=\bigcup_{i}X_{i},\\|\mathrm{P}|\neq1\end{array}$}}
        \sum\limits_{\mbox{\scriptsize$\begin{array}{c}{Z_{1}\subset X_{1}},\\Z_{1}\neq\emptyset\end{array}$}}
        \ldots \sum\limits_{\mbox{\scriptsize$\begin{array}{c}{Z_{|\mathrm{P}|}\subset X_{|\mathrm{P}|}},\\Z_{|\mathrm{P}|}\neq\emptyset\end{array}$}}
        \big(-\mathcal{N}_{\mathrm{int}}^{(\sum\limits_{r=1}^{|\mathrm{P}|}|Z_{{r}}|)}(Z_{{1}},\ldots,Z_{{|\mathrm{P}|}})\big)
        \mathcal{S}^{\pm}_{n}\prod_{X_{i}\subset \mathrm{P}}f_{|X_{i}|}(X_{i}),\nonumber
\end{eqnarray}
\end{theorem}
\begin{proof}
Mapping (\ref{groupKum}) is defined for $f_{n}\in\mathfrak{L}^{1}(\mathcal{H}^{\pm}_{n})$, $n\geq1$,
and the following inequality holds
\begin{eqnarray}\label{ocinka}
     &&\big\|\big(\mathfrak{A}_{t}(f)\big)_{n}\big\|_{\mathfrak{L}^{1}(\mathcal{H}^{\pm}_{n})}\leq n!e^{3n}c^{n},
\end{eqnarray}
where $c\equiv \max(\widetilde{c},1)$ and $\widetilde{c}\equiv \max_{\mathrm{P}:\,(1,\ldots,n)=\bigcup_iX_i}\|f_{|X_{i}|}\|_{_{\mathfrak{L}^{1}(\mathcal{H}^{\pm}_{|X_{i}|})}}.$
Indeed, since for $f_{n}\in \mathfrak{L}^{1}(\mathcal{H}^{\pm}_{n})$ it holds (Proposition 1)
\begin{eqnarray*}
     &&\mathrm{Tr}_{1,\ldots,n}|\mathcal{G}_{n}(-t)f_{n}|=\big\|f_{n}\big\|_{\mathfrak{L}^{1}(\mathcal{H}^{\pm}_{n})},\qquad n\geq1,
\end{eqnarray*}
hence taking into account the inequality
\begin{eqnarray*}
     &&\big\|\mathcal{S}^{\pm}_{n}\prod_{X_{i}\subset \mathrm{P}}f_{|X_{i}|}\big\|_{_{\mathfrak{L}^{1}(\mathcal{H}^{\pm}_{n})}}
        \leq\prod_{X_{i}\subset \mathrm{P}}\big\|f_{|X_{i}|}\big\|_{_{\mathfrak{L}^{1}(\mathcal{H}^{\pm}_{|X_{i}|})}},
\end{eqnarray*}
we obtain
\begin{eqnarray*}
       &&\big\|\big(\mathfrak{A}_{t}(f)\big)_{n}\big\|_{\mathfrak{L}^{1}(\mathcal{H}^{\pm}_{n})}
            \leq\sum\limits_{\mathrm{P}:\,(1,\ldots,n)=\bigcup_i X_i}\,\,
            \sum\limits_{\mathrm{P}^{'}:\,(\{X_1\},\ldots,\{X_{|\mathrm{P}|}\})=\bigcup_k Z_k}
            (|\mathrm{P}^{'}|-1)!\,\prod_{X_i\subset \mathrm{P}}\big\|f_{|X_i|}\big\|_{\mathfrak{L}^{1}
            (\mathcal{H}^{\pm}_{|X_{i}|})}\leq\\
       &&\leq \sum\limits_{\mathrm{P}:\,(1,\ldots,n)=\bigcup_i X_i} c^{|\mathrm{P}|}
            \sum\limits_{k=1}^{|\mathrm{P}|}\mathrm{s}(|\mathrm{P}|,k)(k-1)!
            \leq \sum\limits_{\mathrm{P}:\,(1,\ldots,n)=\bigcup_i X_i} c^{|\mathrm{P}|}
            \sum\limits_{k=1}^{|\mathrm{P}|}k^{|\mathrm{P}|-1}=\\
       &&=\sum\limits_{m=1}^{n}\mathrm{s}(n,m)c^{m}
            \sum\limits_{k=1}^{m}k^{m-1}\leq n!e^{3n}c^{n},
\end{eqnarray*}
where $\mathrm{s}(|\mathrm{P}|,k)$ are the Stirling numbers of the second kind. Thus, $\big(\mathfrak{A}_{t}(f)\big)_{n}\in \mathfrak{L}^{1}(\mathcal{H}^{\pm}_{n})$ for arbitrary $t\in\mathbb{R}^{1}$ and $n\geq 1$.

We now prove the group property of a one-parametric family of the nonlinear operators $\mathfrak{A}_{t}(\cdot)$
that are defined by (\ref{groupKum}), i.e.
\begin{eqnarray*}
      &&\mathfrak{A}_{t_{1}}\big(\mathfrak{A}_{t_{2}}(f)\big)=
         \mathfrak{A}_{t_{2}}\big(\mathfrak{A}_{t_{1}}(f)\big)=\mathfrak{A}_{t_{1}+t_{2}}(f).
\end{eqnarray*}
For $f_{n}\in \mathfrak{L}^{1}(\mathcal{H}^{\pm}_{n})$, $n\geq1$, and for arbitrary
$t_{1},\,t_{2}\in \mathbb{R}^{1}$, according to definition (\ref{cumulantP}), it holds
\begin{eqnarray*}
      &&\big(\mathfrak{A}_{t_{1}}\big(\mathfrak{A}_{t_{2}}(f)\big)\big)_{n}(1,\ldots,n)=\sum\limits_{\mathrm{P}:\,(1,\ldots,n)=\bigcup_iX_i}
            \mathfrak{A}_{|\mathrm{P}|}(t_{1},\{X_1\},\ldots,\{X_{|\mathrm{P}|}\})\times\\
      &&\times\prod_{X_i\subset \mathrm{P}}\,\sum\limits_{\mathrm{P}_{i}:\,X_{i}=\bigcup_{l_{i}}Z_{l_i}}
            \mathfrak{A}_{|\mathrm{P}_{i}|}(t_{2},\{Z_{1_i}\},\ldots,\{Z_{|\mathrm{P}_i|_i}\})
            \,\mathcal{S}_n^{\pm}\prod_{Z_{l_i}\subset \mathrm{P}_{i}}f_{|Z_{l_i}|}(Z_{l_i})=\\
      &&=\sum\limits_{\mathrm{P}:\,(1,\ldots,n)=\bigcup_i X_i}\,\,
            \sum\limits_{\mathrm{P}^{'}:\,(\{X_1\},\ldots,\{X_{|\mathrm{P}|}\})=\bigcup_j Q_j}
            (-1)^{|\mathrm{P}^{'}|-1}(|\mathrm{P}^{'}|-1)!
            \prod\limits_{Q_{j}\subset\mathrm{P}^{'}}\mathcal{G}_{|\theta(Q_{j})|}(-t_{1},\theta(Q_{j}))\times\\
      &&\times\prod_{X_i\subset \mathrm{P}}\,\sum\limits_{\mathrm{P}_{i}:\,X_{i}=\bigcup_{l_i}Z_{l_i}}
            \,\,\sum\limits_{\mathrm{P}^{'}_{i}:\,(\{Z_{1_i}\},\ldots,\{Z_{|\mathrm{P}_i|_i}\})=\bigcup_{k_i}R_{k_i}}
            (-1)^{|\mathrm{P}^{'}_{i}|-1}(|\mathrm{P}^{'}_{i}|-1)!\times\\
      &&\times\prod\limits_{R_{k_i}\subset\mathrm{P}^{'}_{i}}
            \mathcal{G}_{|\theta(R_{k_i})|}(-t_{2},\theta(R_{k_i}))\,\mathcal{S}_n^{\pm}\prod_{Z_{l_i}\subset
            \mathrm{P}_{i}}f_{|Z_{l_i}|}(Z_{l_i}),
\end{eqnarray*}
where $(\{Z_{1_i}\},\ldots,\{Z_{|\mathrm{P}_i|_i}\})$
is a set, elements of which are $|\mathrm{P}_{i}|$ subsets $Z_{l_i}\subset X_{i}$ of the partition
$\mathrm{P}_{i}:\,X_{i}=\cup_{l_i}Z_{l_i}.$
Collecting similar terms of products of operators $f_{n}$, $n\geq1$, and taking into account
the group property of operators  $\mathcal{G}_{n}(-t)$, $n\geq1$, (\ref{groupG}) we establish
\begin{eqnarray*}
      &&\big(\mathfrak{A}_{t_{1}}\big(\mathfrak{A}_{t_{2}}(f)\big)\big)_{n}(1,\ldots,n)=\\
      &&=\sum\limits_{\mathrm{P}:\,(1,\ldots,n)=\bigcup_i X_i}\,\,
          \sum\limits_{\mathrm{P}^{'}:\,(\{X_1\},\ldots,\{X_{|\mathrm{P}|}\})=\bigcup_k Q_k}
          (-1)^{|\mathrm{P}^{'}|-1}(|\mathrm{P}^{'}|-1)!
          \prod\limits_{Q_{k}\subset\mathrm{P}^{'}}\mathcal{G}_{|\theta(Q_{k})|}(-t_{1}-t_{2},\theta(Q_{k})) \times\\
      &&\times\mathcal{S}_n^{\pm}\prod_{X_i\subset \mathrm{P}}f_{|X_i|}(X_i)=\sum\limits_{\mathrm{P}:\,(1,\ldots,n)=\bigcup_i X_i}
          \mathfrak{A}_{|\mathrm{P}|}(t_{1}+t_{2},\{X_1\},\ldots,\{X_{|\mathrm{P}|}\})
          \,\mathcal{S}_n^{\pm}\prod_{X_i\subset \mathrm{P}}f_{|X_i|}(X_i)=\\
      &&=\big(\mathfrak{A}_{t_{1}+t_{2}}(f)\big)_{n}(1,\ldots,n).
\end{eqnarray*}

The strong continuity property of the group $\mathfrak{A}_t(\cdot)$ over the parameter $t\in \mathbb{R}^{1}$
is the straightforward consequence of the strong continuity of groups (\ref{groupG}) of the von Neumann equations \cite{DauL_5}.
Indeed, according to identity (\ref{Stirl}) the following  equality is valid
\begin{eqnarray*}
       &&\sum\limits_{\mathrm{P}:\,(1,\ldots,n)=\bigcup_i X_i}\,\,
           \sum\limits_{\mathrm{P}^{'}:\,(\{X_1\},\ldots,\{X_{|\mathrm{P}|}\})=\bigcup_k Z_k}
           (-1)^{|\mathrm{P}^{'}|-1}(|\mathrm{P}^{'}|-1)!\,\mathcal{S}_n^{\pm}
           \prod_{X_i\subset \mathrm{P}}f_{|X_i|}(X_i)= f_{n}(1,\ldots,n).
\end{eqnarray*}
Therefore, for $f_{n}\in\mathfrak{L}^{1}_{0}(\mathcal{H}^{\pm}_{n}) \subset \mathfrak{L}^{1}(\mathcal{H}^{\pm}_{n})$,
$n\geq 1$, it holds
\begin{eqnarray*}
       &&\lim_{t\rightarrow 0}\big\|\sum\limits_{\mathrm{P}:\,(1,\ldots,n)=\bigcup_iX_i}\,\,
            \sum\limits_{\mathrm{P}^{'}:\,(\{X_1\},\ldots,\{X_{|\mathrm{P}|}\})=\bigcup_k Z_k}
            (-1)^{|\mathrm{P}^{'}|-1}(|\mathrm{P}^{'}|-1)!\prod\limits_{Z_{k}\subset\mathrm{P}^{'}}
            \mathcal{G}_{|\theta(Z_{k})|}(-t,\theta(Z_{k}))\times\\
       &&\times \mathcal{S}_n^{\pm}\prod_{X_i\subset \mathrm{P}}f_{|X_i|}(X_i)-f_{n}\big\|_{\mathfrak{L}^{1}(\mathcal{H}^{\pm}_{n})}
            \leq\sum\limits_{\mathrm{P}:\,(1,\ldots,n)=\bigcup_iX_i}\,\,
            \sum\limits_{\mathrm{P}^{'}:\,(\{X_1\},\ldots,\{X_{|\mathrm{P}|}\})=\bigcup_k Z_k}(|\mathrm{P}^{'}|-1)!\times\\
       &&\times\lim_{t\rightarrow 0}\big\|\prod\limits_{Z_{k}\subset\mathrm{P}^{'}}\mathcal{G}_{|\theta(Z_{k})|}(-t,\theta(Z_{k}))\,
            \mathcal{S}_n^{\pm}\prod_{X_i\subset \mathrm{P}}f_{|X_i|}(X_i)
            -\mathcal{S}_n^{\pm}\prod_{X_i\subset \mathrm{P}}f_{|X_i|}(X_i)\big\|_{\mathfrak{L}^{1}(\mathcal{H}^{\pm}_{n})}.
\end{eqnarray*}
In view of the fact that the mapping $\mathcal{G}_{n}(-t)$ is strong continuous group,
i.e. in the sense of the norm convergence in $\mathfrak{L}^{1}(\mathcal{H}^{\pm}_{n})$ there exists the limit
\begin{eqnarray*}
   &&\lim\limits_{t\rightarrow0}(\mathcal{G}_{n}(-t)f_{n}-f_{n})=0,
\end{eqnarray*}
this implyies that for mutually disjoint subsets  $X_{i}\subset Y,$ the following equality is also valid
\begin{eqnarray*}
   &&\lim_{t\rightarrow 0}\big(\prod\limits_{Z_{k}\subset\mathrm{P}^{'}}\mathcal{G}_{|\theta(Z_{k})|}(-t,\theta(Z_{k}))f_{n}-f_{n}\big)=0.
\end{eqnarray*}
For $f_{n}\in\mathfrak{L}_{0}^{1}(\mathcal{H}^{\pm}_{n})\subset\mathfrak{L}^{1}(\mathcal{H}^{\pm}_{n}),$ we finally establish
\begin{eqnarray*}
   &&\lim_{t\rightarrow 0}\big\|\big(\mathfrak{A}_{t}(f)\big)_{n}-f_{n}\big\|_{\mathfrak{L}^{1}(\mathcal{H}^{\pm}_{n})}=0.
\end{eqnarray*}

To construct the generator $\mathfrak{N}(\cdot)$ of strong continuous group (\ref{groupKum})
we first differentiate it in the sense of the pointwise convergence in the spaces
$\mathfrak{L}^{1}(\mathcal{H}^{\pm}_{n})$, i.e.
for arbitrary $\psi_{n}\in \mathcal{D}(H_{n})\subset\mathcal{H}^{\pm}_{n}$.
Taking into consideration the fact that for
$f_{n}\in\mathfrak{L}_{0}^{1}(\mathcal{H}^{\pm}_{n})\subset\mathcal{D}(\mathfrak{N}(\cdot)_{n})$ equalities (\ref{infOper})
and (\ref{oper Nint1}) hold, for group (\ref{groupKum}) we derive
\begin{eqnarray}\label{derivation}
       &&\lim\limits_{t\rightarrow 0}\frac{1}{t}\big((\mathfrak{A}_{t}(f))_{n}-f_{n}\big)\psi_{n}=\\
       &&=\lim\limits_{t\rightarrow 0}\frac{1}{t}\big(\sum\limits_{\mathrm{P}:(1,\ldots,n)=\bigcup_i X_i}
            \mathfrak{A}_{|\mathrm{P}|}(t,\{X_1\},\ldots,\{X_{|\mathrm{P}|}\})\,
            \mathcal{S}_n^{\pm}\prod_{X_i\subset \mathrm{P}}f_{|X_i|}(X_i)-f_{n}\big)\psi_{n}=\nonumber\\
       &&=\lim\limits_{t\rightarrow 0}\frac{1}{t}\big(\mathfrak{A}_{1}(t,\{1,\ldots,n\})f_{n}-f_{n}\big)\psi_{n}+\nonumber\\
       &&+\sum\limits_{\mbox{\scriptsize $\begin{array}{c}{\mathrm{P}}:(1,\ldots,n)=\bigcup_i X_i,\\|{\mathrm{P}}|>1\end{array}$}}
            \lim\limits_{t\rightarrow 0}\frac{1}{t}\mathfrak{A}_{|\mathrm{P}|}(t,\{X_1\},\ldots,\{X_{|\mathrm{P}|}\})\,
            \mathcal{S}_n^{\pm}\prod_{X_i\subset \mathrm{P}}f_{|X_i|}(X_i)\psi_{n}=\nonumber\\
       &&=(-\mathcal{N}_{n}f_{n})\psi_{n}
            +\hskip-2mm\sum\limits_{\mbox{\scriptsize $\begin{array}{c}{\mathrm{P}}:(1,\ldots,n)=\bigcup_iX_i,\\|{\mathrm{P}}|>1\end{array}$}}
            \hskip-2mm\big(-\mathcal{N}^{\mathrm{int}}(\{X_1\},\ldots,\{X_{|\mathrm{P}|}\})\big)\,\mathcal{S}_n^{\pm}
            \prod_{X_i\subset \mathrm{P}}f_{|X_i|}(X_i)\psi_{n},\nonumber
\end{eqnarray}
where the operator $\mathcal{N}^{\mathrm{int}}$ is given by the formula
\begin{eqnarray}\label{gener Nnl}
       &&\mathcal{N}^{\mathrm{int}}(\{X_1\},\ldots,\{X_{|\mathrm{P}|}\})\equiv
            \sum\limits_{\mbox{\scriptsize$\begin{array}{c}{Z_{1}\subset X_{1}},\\Z_{1}\neq\emptyset \end{array}$}}\ldots
            \sum\limits_{\mbox{\scriptsize
            $\begin{array}{c}{Z_{|\mathrm{P}|}\subset X_{|\mathrm{P}|}},\\Z_{|\mathrm{P}|}\neq\emptyset\end{array}$}}
            \mathcal{N}_{\mathrm{int}}^{(\sum\limits_{r=1}^{|\mathrm{P}|}|Z_{{r}}|)}
            (Z_{{1}},\ldots,Z_{{|\mathrm{P}|}}).
\end{eqnarray}
Hence in view of equality (\ref{derivation})
for $f_{n}\in\mathfrak{L}_{0}^{1}(\mathcal{H}^{\pm}_{n})\subset\mathcal{D}
(\mathfrak{N}(\cdot)_n)\subset\mathfrak{L}^{1}(\mathcal{H}^{\pm}_{n}),\, n\geq1,$ in the sense of the norm convergence
in $\mathfrak{L}^{1}(\mathcal{H}^{\pm}_{n}),$  we finally establish
\begin{eqnarray*}
   &&\lim_{t\rightarrow 0}\big\| \frac{1}{t}\big((\mathfrak{A}_{t}(f))_{n}- f_{n}\big)
       -\big(\mathfrak{N}(f)\big)_{n}\big\|_{\mathfrak{L}^{1}(\mathcal{H}^{\pm}_{n})}=0,
\end{eqnarray*}
where $\mathfrak{N}(\cdot)$ is given by formula (\ref{Nnl}).
\end{proof}

\subsection{The existence and uniqueness theorem}
Let $Y=(1,\ldots,s)$ is above accepted notation. For abstract initial-value problem
(\ref{vonNeumanHierarchyAlternateDescriptionFB}),(\ref{vonNeumanHierarchyAlternateDescriptionInitialValue})
in the spaces $\mathfrak{L}^{1}(\mathcal{F}^{\pm}_\mathcal{H})$ of sequences of trace class operators the following theorem is true.
\begin{theorem}
The solution of initial-value problem (\ref{vonNeumanHierarchyAlternateDescriptionFB}),(\ref{vonNeumanHierarchyAlternateDescriptionInitialValue})
of the von Neumann hierarchy
is determined by the following expansion
\begin{eqnarray*}
    && g_{s}(t,Y)=\sum\limits_{\mathrm{P}:\,Y=\bigcup_iX_i}
        \mathfrak{A}_{|\mathrm{P}|}(t,\{X_1\},\ldots,\{X_{|\mathrm{P}|}\})\,
        \mathcal{S}_s^{\pm}\prod_{X_i\subset \mathrm{P}}g_{|X_i|}(0,X_i),\quad s\geq1,
\end{eqnarray*}
where $\mathfrak{A}_{|\mathrm{P}|}(t)$ is the $|\mathrm{P}|$-order cumulant of the groups of operators (\ref{groupG})
defined by formula (\ref{cumulantP}).

For $g_{n}(0)\in \mathfrak{L}^{1}_{0}(\mathcal{H}^{\pm}_{n})\subset
\mathfrak{L}^{1}(\mathcal{H}^{\pm}_{n})$ one is a strong (classical) solution and for arbitrary initial data $g_{n}(0) \in
\mathfrak{L}^{1}(\mathcal{H}^{\pm}_{n})$ one is a weak (generalized) solution.
\end{theorem}
\begin{proof} The Theorem 1 implyies that for initial data
$g_{n}(0)\in\mathfrak{L}_{0}^{1}(\mathcal{H}^{\pm}_{n})\subset\mathfrak{L}^{1}(\mathcal{H}^{\pm}_{n}),\, n\geq1,$ sequence (\ref{rozvNF-N_F})
is a strong solution of initial-value problem
(\ref{vonNeumanHierarchyAlternateDescriptionFB}),(\ref{vonNeumanHierarchyAlternateDescriptionInitialValue}).

Let us prove that in case of arbitrary initial data $g(0)\in\mathfrak{L}^{1}(\mathcal{F}^{\pm}_\mathcal{H})$
expansion (\ref{rozvNF-N_F}) is a weak solution of
the initial-value problem of the von Neumann hierarchy (\ref{vonNeumanHierarchyAlternateDescriptionFB}).
With this purpose we introduce the functional
\begin{eqnarray}\label{func-g}
    &&\big(f_{s},g_{s}(t)\big)\doteq
       \mathrm{Tr}_{1,\ldots,s}\,f_{s}(1,\ldots,s)\,g_{s}(t,1,\ldots,s),
\end{eqnarray}
where $f_{s}\in \mathfrak{L}_{0}(\mathcal{H}^{\pm}_{s})$  is degenerate bounded operator with infinitely times differentiable
kernel with compact support and the operator $g_{s}(t)$ is defined by expansion (\ref{rozvNF-N_F}).
According to estimate (\ref{ocinka}), for $f_{s}\in \mathfrak{L}_{0}(\mathcal{H}^{\pm}_{s})$, functional (\ref{func-g}) exists.

Using expansion (\ref{rozvNF-N_F}), we transform  functional (\ref{func-g}) as follows
\begin{eqnarray}\label{funk-gN}
    &&\big(f_{s},g_{s}(t)\big)=\mathrm{Tr}_{1,\ldots,s}\sum\limits_{\mathrm{P}:Y=\bigcup_iX_i}f_{s}\,
          \mathfrak{A}_{|\mathrm{P}|}(t,\{X_1\},\ldots,\{X_{|\mathrm{P}|}\})
          \mathcal{S}_s^{\pm}\prod_{X_i\subset \mathrm{P}}g_{|X_i|}(0,X_i)=\\ \nonumber
    &&=\mathrm{Tr}_{1,\ldots,s}\,\sum\limits_{\mathrm{P}:\,Y=\bigcup_iX_i}\,\,
          \sum\limits_{\mathrm{P}^{'}:\,(\{X_1\},\ldots,\{X_{|\mathrm{P}|}\})=\bigcup_{k}Z_{k}}
          (-1)^{|\mathrm{P}^{'}|-1}(|\mathrm{P}^{'}|-1)!\times\\ \nonumber
    &&\times\prod\limits_{Z_{k}\subset\mathrm{P}^{'}}
          \mathcal{G}_{|\theta(Z_{k})|}(t,\theta(Z_{k}))f_{s}\,
          \mathcal{S}_s^{\pm}\prod\limits_{X_{i}\subset\mathrm{P}}g_{|X_{i}|}(0,X_{i})\nonumber,
\end{eqnarray}
where the group of operators
$\mathcal{G}_{|\theta(Z_{k})|}(t)$ is adjoint to the group $\mathcal{G}_{|\theta(Z_{k})|}(-t)$ in the sense of functional (\ref{func-g}).

For $g_{n}(0)\in\mathfrak{L}^{1}(\mathcal{H}^{\pm}_{n})$ and
$f_{n}\in \mathfrak{L}_{0}(\mathcal{H}^{\pm}_{n})$ within the context of the Proposition 1 \cite{BR} the following equality holds
\begin{eqnarray}
   &&\lim\limits_{\triangle t\rightarrow0}\mathrm{Tr}_{1,\ldots,n}\,
       \big(\frac{1}{\triangle t}(\mathcal{G}_{n}(t+\triangle t)f_{n}-\mathcal{G}_{n}(t)f_{n})
       -\mathcal{G}_{n}(t)\mathcal{N}_{n}f_{n}\big)g_{n}(0)=0,\nonumber
\end{eqnarray}
and, therefore we have
\begin{eqnarray*}
    &&\lim\limits_{\triangle t\rightarrow0}\mathrm{Tr}_{1,\ldots,s}\,\frac{1}{\triangle t}
        \big(\prod\limits_{Z_{k}\subset\mathrm{P}^{'}}
        \mathcal{G}_{|\theta(Z_{k})|}(t+\triangle t,\theta(Z_{k}))f_{s}(Y)-
        \prod\limits_{Z_{k}\subset\mathrm{P}^{'}}\mathcal{G}_{|\theta(Z_{k})|}(t,\theta(Z_{k})f_{s}(Y)\big)\times\\
    &&\times\mathcal{S}_s^{\pm}\prod\limits_{X_{i}\subset\mathrm{P}}g_{|X_{i}|}(0,X_{i})=\\
    &&=\mathrm{Tr}_{1,\ldots,s}\,\sum\limits_{Z_{l}\subset\mathrm{P}^{'}}
        \mathcal{N}_{|\theta(Z_{l})|}(\theta(Z_{l}))f_{s}(Y)\,
        \prod\limits_{Z_{k}\subset\mathrm{P}^{'}}\mathcal{G}_{|\theta(Z_{k})|}(-t,\theta(Z_{k}))
        \mathcal{S}_s^{\pm}\prod\limits_{X_{i}\subset\mathrm{P}}g_{|X_{i}|}(0,X_{i}).
\end{eqnarray*}
For $f_{s}\in \mathfrak{L}_{0}(\mathcal{H}^{\pm}_{s})$
and bounded interaction potentials the limit functional exists.
Then, taking into account that symmetrization and anti-symmetrization operators (\ref{Sn}) are integrals of motion,
the derivative over time variable of functional (\ref{func-g}) represented by (\ref{funk-gN}) is given by
\begin{eqnarray*}
    &&\frac{d}{dt}\big(f_{s},\,g_{s}(t)\big)=\big(\mathcal{N}_{s}f_{s},\,g_{s}(t)\big)+
         \mathrm{Tr}_{1,\ldots,s}\,\sum\limits_{\mbox{\scriptsize $\begin{array}{c}{\mathrm{P}}:\,Y=\bigcup_i X_i,\\|{\mathrm{P}}|>1\end{array}$}}
         \sum\limits_{\mathrm{P}^{'}:\,(\{X_1\},\ldots,\{X_{|\mathrm{P}|}\})=\bigcup_{k}Z_{k}}
         (-1)^{|\mathrm{P}^{'}|-1}\times\\
    &&\times(|\mathrm{P}^{'}|-1)!\sum\limits_{Z_{l}\subset\mathrm{P}^{'}}
         \mathcal{N}_{|\theta(Z_{l})|}(\theta(Z_{l}))f_{s}(Y)\,
         \prod\limits_{Z_{k}\subset\mathrm{P}^{'}}\mathcal{G}_{|\theta(Z_{k})|}(-t,\theta(Z_{k}))\,
         \mathcal{S}_s^{\pm}\prod\limits_{X_{i}\subset\mathrm{P}}g_{|X_{i}|}(0,X_{i})=\\
    &&=\big(\mathcal{N}_{s}f_{s},\,g_{s}(t)\big)+
         \mathrm{Tr}_{1,\ldots,s}\,\sum\limits_{\mbox{\scriptsize $\begin{array}{c}{\mathrm{P}}:\,Y=\bigcup_iX_i,\\|{\mathrm{P}}|>1\end{array}$}}
         \mathcal{N}^{\mathrm{int}}(\{X_1\},\ldots,\{X_{|\mathrm{P}|}\})f_{s}(Y)\,
         \mathcal{S}_s^{\pm}\prod_{X_i\subset \mathrm{P}}g_{|X_i|}(t,X_i),
\end{eqnarray*}
where the operator $\mathcal{N}^{\mathrm{int}}(\{X_1\},\ldots,\{X_{|\mathrm{P}|}\})$
is defined by formulae (\ref{gener Nnl}) and (\ref{oper Nint2}).

Thus, the result of the differentiation over time variable of functional (\ref{func-g}) takes the form
\begin{eqnarray}\label{d_funk-gN}
    &&\frac{d}{dt}\big(f_{s},\,g_{s}(t)\big)=\big(\mathcal{N}_{s}f_{s},\,g_{s}(t)\big)+\\
    &&+\mathrm{Tr}_{1,\ldots,s}\,\sum\limits_{\mbox{\scriptsize $\begin{array}{c}{\mathrm{P}}:\,Y=\bigcup_iX_i,\\|{\mathrm{P}}|>1\end{array}$}}
        \mathcal{N}^{\mathrm{int}}(\{X_1\},\ldots,\{X_{|\mathrm{P}|}\})f_{s}(Y)\,\mathcal{S}_s^{\pm}
        \prod_{X_i\subset \mathrm{P}}g_{|X_i|}(t,X_i).\nonumber
\end{eqnarray}
Equality (\ref{d_funk-gN}) means that for arbitrary initial data
$g_{n}(0)\in\mathfrak{L}^{1}(\mathcal{H}^{\pm}_{n}),$ $n\geq 1,$ a weak solution of the initial-value problem
of the von Neumann hierarchy (\ref{vonNeumanHierarchyAlternateDescriptionFB}),(\ref{vonNeumanHierarchyAlternateDescriptionInitialValue})
is determined by formula (\ref{rozvNF-N_F}).
\end{proof}

The statement of this theorem also holds for the correlation operators of clusters of particles introduced in section 2.5.

\section{On applications of the von Neumann hierarchy}

\subsection{Functionals of the mean-value and the dispersion of observables}
Observables of the Bose and Fermi many-particle systems are sequences $A=(A_{0},A_{1},\ldots,$ $A_{n},\ldots)$ of self-adjoint
operators defined on the Fock spaces $\mathcal{F}^{\pm}_{\mathcal{H}}$ and $A_{0}\in \mathbb{C}$.
We consider observables as elements of the spaces of sequences of bounded operators
$\mathfrak{L}(\mathcal{F}^{\pm}_\mathcal{H})$ with an operator norm \cite{DauL_5,Pe95}.
The average values of observables (mean values of observables) are defined by the positive continuous
linear functional on the spaces $\mathfrak{L}(\mathcal{F}^{\pm}_\mathcal{H})$
\begin{eqnarray}\label{averageD}
     &&\langle A\rangle=(I,D)^{-1}\sum\limits_{n=0}^{\infty}\frac{1}{n!}
         \,\mathrm{Tr}_{1,\ldots,n}\,A_{n}\,D_{n},
\end{eqnarray}
where $\mathrm{Tr}_{1,\ldots,n}$ are the partial traces over $1,\ldots,n$ particles,
$(I,D)={\sum\limits}_{n=0}^{\infty}\frac{1}{n!}\mathrm{Tr}_{1,\ldots,n}D_{n}$ is a normalizing
factor (grand canonical partition function) and $D=(1,D_{1},\ldots,D_{n},\ldots)$ is
a sequence of self-adjoint positive density operators defined on the Fock spaces $\mathcal{F}^{\pm}_{\mathcal{H}}$ that
describes the states of a quantum system of non-fixed number of particles.
For $D\in \mathfrak{L}^{1} (\mathcal{F}^{\pm}_\mathcal{H})$
and $A\in \mathfrak{L}(\mathcal{F}^{\pm}_\mathcal{H})$ average value functional (\ref{averageD}) exists.
It determines a duality between observables and states.

We remark that in case of a system of fixed number $N$ of particles the observables
and states are one-component sequences, respectively, $A^{(N)}=(0,\ldots,0,A_{N},0,\ldots),$
$D^{(N)}=(0,\ldots,$ $0,D_{N},0,\ldots)$, therefore, formula (\ref{averageD}) for an average value reduces to
\begin{eqnarray*}
   &&\langle A^{(N)}\rangle=(\mathrm{Tr}_{1,\ldots,N}D_{N})^{-1}\mathrm{Tr}_{1,\ldots,N}A_{N}D_{N}.
\end{eqnarray*}

In the framework of the description of states by correlation operators (\ref{gfromDFB}) the average values, for example,
of additive-type observables $A^{(1)}=(0,a_{1},\ldots,{\sum\limits}_{i=1}^{n}a_{1}(i),\ldots)$ are given by the functional
\begin{eqnarray}\label{averageg}
      &&\langle A^{(1)}\rangle(t)=\sum\limits_{n=0}^{\infty}\frac{1}{n!}
         \,\mathrm{Tr}_{1,\ldots,1+n}\,a_{1}(1)g_{1+n}(t,1,2,\ldots,1+n),
\end{eqnarray}
or in general case of $s$-particle observables $A^{(s)}=(0,\ldots,0,a_{s}(1,\ldots,s),\ldots,
\sum_{i_{1}<\ldots<i_{s}=1}^{n}a_s(i_{1},\ldots,$ $i_{s}),\ldots)$ by the functional
\begin{eqnarray}\label{averagegs}
      &&\langle A^{(s)}\rangle(t)=\frac{1}{s!}\sum\limits_{n=0}^{\infty}\frac{1}{n!}
         \,\mathrm{Tr}_{1,\ldots,s+n}\,a_{s}(1,\ldots,s)g_{1+n}(t,\{1,\ldots,s\},s+1,\ldots,s+n),
\end{eqnarray}
where the correlation operator $g_{1+n}(t)$ is defined by expansion (\ref{rozvNF-N_F_clusters}).
For $A^{(s)}\in \mathfrak{L}(\mathcal{F}^{\pm}_\mathcal{H})$ and $g^{(s)}\in \mathfrak{L}^{1}(\mathcal{F}^{\pm}_\mathcal{H})$
functional (\ref{averagegs}) exists.

In view of the definition of average value functional (\ref{averagegs}), the dispersion of an additive-type observable is
defined by the functional
\begin{eqnarray}\label{dispg}
     &&\langle(A^{(1)}-\langle A^{(1)}\rangle(t))^2\rangle(t)=\sum\limits_{n=0}^{\infty}\frac{1}{n!}
         \,\mathrm{Tr}_{1,\ldots,1+n}\,(a_1^2(1)-\langle A^{(1)}\rangle^2(t)) g_{1+n}(t,1,\ldots,1+n) +\nonumber\\
     &&+\sum\limits_{n=0}^{\infty}\frac{1}{n!}
         \,\mathrm{Tr}_{1,\ldots,2+n}\,a_{1}(1)a_{1}(2)g_{2+n}(t,1,\ldots,2+n),
\end{eqnarray}
where $\langle A^{(1)}\rangle(t)$ is determined by expression (\ref{averageg}).
For $A^{(1)}\in \mathfrak{L}(\mathcal{F}^{\pm}_\mathcal{H})$ and $g\in \mathfrak{L}^{1} (\mathcal{F}^{\pm}_\mathcal{H})$
functional (\ref{dispg}) exists.

To justify definition (\ref{averagegs}), i.e. to derive functional (\ref{averagegs}) from (\ref{averageD}), 
we introduce necessary notions and formulate some equalities.

For arbitrary $f=(f_{0},f_{1},\ldots,f_{n},\ldots)\in\mathfrak{L}^{1}(\mathcal{F}^{\pm}_\mathcal{H})$ and
$Y\equiv(1,\ldots,s)$ we define the linear mapping $\mathfrak{d}_{Y}:f\rightarrow \mathfrak{d}_{Y}f$
by the formula
\begin{eqnarray}\label{oper_d}
    &&(\mathfrak{d}_{Y} f)_{n}\doteq f_{|Y|+n}(Y,|Y|+1,\ldots,|Y|+n),\quad n\geq0.
\end{eqnarray}
For the set $\{Y\}$ defined above in section 2.4 we have respectively
\begin{eqnarray}\label{oper_c}
    &&(\mathfrak{d}_{\{Y\}} f)_{n}\doteq f_{1+n}(\{Y\},s+1,\ldots,s+n),\quad n\geq0.
\end{eqnarray}
On sequences $\mathfrak{d}_{Y}f$ and $\mathfrak{d}_{Y'}\widetilde{f}$ we introduce the $\circledast$-product
\begin{eqnarray}\label{roundStarProduct}
    &&(\mathfrak{d}_{Y}f\circledast\mathfrak{d}_{Y'}\widetilde{f})_{|X|}(X)\doteq
        \sum\limits_{Z\subset X}\mathcal{S}^{\pm}_{|\theta(Y \bigcup Y'\bigcup X)|}f_{|Z|+|Y|}(Y,Z)
        \,\widetilde{f}_{|X\backslash Z|+|Y'|}(Y',X\backslash Z),
\end{eqnarray}
where $\mathcal{S}^{\pm}_{|\cdot|}$ are operators defined by formula (\ref{Sn}), $\theta$
is the declasterization mapping defined by (\ref{Theta}),
$X,Y,Y'$ are the sets, elements of which characterize clusters of particles,
$\sum_{Z\subset X}$ is the sum over all subsets $Z$ of the set $X$.
In particular case ($Y=\emptyset$, $Y'=\emptyset$) definition (\ref{roundStarProduct}) reduces to (\ref{Product}).

For $f=(0,f_{1},\ldots,f_{n},\ldots)$,  $f_{n}\in\mathfrak{L}^{1}(\mathcal{H}^{\pm}_{n})$,
according to definitions of mappings (\ref{circledExp}) and (\ref{oper_c}) the following equality holds
\begin{eqnarray}\label{d_gamma}
      &&\mathfrak{d}_{\{Y\}}\mathbb{E}\mathrm{xp}_{\circledast} f=
      \mathbb{E}\mathrm{xp}_{\circledast}f\circledast\mathfrak{d}_{\{Y\}}f,
\end{eqnarray}
and for mapping (\ref{oper_d}) correspondingly,
\begin{eqnarray}\label{oper_dex}
      &&\mathfrak{d}_{Y}\mathbb{E}\mathrm{xp}_{\circledast}f=
          \mathbb{E}\mathrm{xp}_{\circledast} f\circledast\sum\limits_{\mathrm{P}:\,Y=\bigcup_i X_{i}}
          \mathfrak{d}_{X_1}f\circledast\ldots\circledast \mathfrak{d}_{X_{|\mathrm{P}|}}f,
\end{eqnarray}
where ${\sum\limits}_{\mathrm{P}:\,Y=\bigcup_i X_{i}}$ is the sum over all possible partitions $\mathrm{P}$
of the set $Y\equiv(1,\ldots,s)$ into $|\mathrm{P}|$ nonempty mutually disjoint subsets $X_i\subset Y$.

For $f_{n}\in\mathfrak{L}^{1}(\mathcal{H}^{\pm}_{n})$ an analogue of the annihilation operator is defined by the formula
\begin{eqnarray}\label{a}
    &&(\mathfrak{a}f)_{s}(1,\ldots,s)\doteq \mathrm{Tr}_{s+1}f_{s+1}(1,\ldots,s,s+1),
\end{eqnarray}
and, therefore
\begin{eqnarray*}
     &&(e^{\mathfrak{a}}f)_{s}(1,\ldots,s)=\sum\limits_{n=0}^{\infty}\frac{1}{n!}\mathrm{Tr}_{s+1,\ldots,{s+n}}
         f_{s+n}(1,\ldots,s+n),\quad s\geq0.
\end{eqnarray*}
According to definitions (\ref{a}) and (\ref{roundStarProduct}), for sequences
$f,\widetilde{f}\in\mathfrak{L}^{1}(\mathcal{F}^{\pm}_\mathcal{H})$ the equality holds
\begin{eqnarray}\label{efg}
     &&(e^{\mathfrak{a}}(f\circledast \widetilde{f}))_{0}=
         (e^{\mathfrak{a}}f)_{0}(e^{\mathfrak{a}}\widetilde{f})_{0}.
\end{eqnarray}
In terms of mappings (\ref{oper_d}) and (\ref{oper_c}) generalized cluster expansions (\ref{gce})
take the form
\begin{eqnarray}\label{gcea}
        && \mathfrak{d}_{Y}D(t)=\mathfrak{d}_{\{Y\}}{\mathbb E}\mathrm{xp}_{\circledast}\,\,g(t),
\end{eqnarray}

We now deduce functional (\ref{averagegs}) from the definition of functional (\ref{averageD}), i.e.
\begin{eqnarray*}
     &&\langle A^{(s)}\rangle(t)=\frac{1}{s!}\,(e^\mathfrak{a}D(0))^{-1}_{0}\,
        \mathrm{Tr}_{Y}\,a_{s}(Y)(e^\mathfrak{a}\,\mathfrak{d}_{Y}D(t))_{0},
\end{eqnarray*}
where $\mathrm{Tr}_{Y}\equiv\mathrm{Tr}_{1,\ldots,s}$.
Indeed, using generalized cluster expansions (\ref{gcea}) and as a consequence of equalities (\ref{d_gamma}) and (\ref{efg}), we find
\begin{eqnarray*}
    &&(e^\mathfrak{a}D(0))^{-1}_{0}
        (e^\mathfrak{a}\,\mathfrak{d}_{Y}D(t))_{0}=(e^\mathfrak{a}D(0))^{-1}_{0}
        (e^\mathfrak{a}\,\mathfrak{d}_{\{Y\}}{\mathbb E}\mathrm{xp}_{\circledast}\,\,g(t))_{0}=\\
    &&=(e^\mathfrak{a}D(0))^{-1}_{0}
        (e^\mathfrak{a}\,\mathbb{E}\mathrm{xp}_{\circledast}g(t)\circledast\mathfrak{d}_{\{Y\}}g(t))_{0}
        =(e^\mathfrak{a}D(0))^{-1}_{0}(e^\mathfrak{a}\,\mathbb{E}\mathrm{xp}_{\circledast}g(t))_{0}
        (e^\mathfrak{a}\,\mathfrak{d}_{\{Y\}}g(t))_{0}.
\end{eqnarray*}
Observing that  $\mathbb{E}\mathrm{xp}_{\circledast}g(t)=D(t)$ and
$(e^\mathfrak{a}D(0))_{0}=(e^\mathfrak{a}D(t))_{0}$,
as a final result we derive representation (\ref{averagegs}).

We note that the possibility of the description of states of systems in the framework of correlations arises naturally
as a result of dividing of the series in expression (\ref{averageD}) by the normalizing factor series.
Indeed, the solution of generalized cluster expansions (\ref{gcea}) has the form
\begin{eqnarray*}
    &&\mathfrak{d}_{\{Y\}}g(t)=\,\mathfrak{d}_{\{Y\}}\vartheta\,\mathbb{L}\mathrm{n}_{\circledast}\,D(t).
\end{eqnarray*}
Here for ${f}_{|Y_P|}(Y_P)\in\mathfrak{L}^{1}(\mathcal{H}^{\pm}_{|Y|})$ the mapping $\vartheta$ is introduced by
\begin{eqnarray*}
    (\vartheta f)_{|Y_P|}(Y_P)\doteq f_{|\theta(Y_P)|}(\theta(Y_P))= f_{|Y|}(Y),
\end{eqnarray*}
where $\theta$ is the declasterization mapping (\ref{Theta}) defined on sets $Y_\mathrm{P}\equiv(\{X_1\},\ldots,\{X_{|\mathrm{P}|}\})$
consisting of clusters of particles.
Then in terms of a sequence of the density operators $D(t)$ the following representation of functional (\ref{averageD}) holds
\begin{eqnarray*}
    &&\langle A^{(s)}\rangle(t)=\frac{1}{s!}\,\mathrm{Tr}_{Y}\,a_{s}(Y)
       (e^\mathfrak{a}\,\mathfrak{d}_{\{Y\}}\vartheta\,\mathbb{L}\mathrm{n}_{\circledast}\,D(t))_{0}.
\end{eqnarray*}

It should be emphasized that correlation operators that belong to the spaces
$\mathfrak{L}^{1}(\mathcal{F}^{\pm}_\mathcal{H})$ describe only finitely many particles, i.e. systems
with finite average number of particles. Indeed, for the observable of number
of particles $N=(0,I,2I,\ldots,nI,\ldots)$ functional (\ref{averageg}) has the form
\begin{eqnarray*}
     &&\langle N\rangle(t)=\sum\limits_{n=0}^{\infty}
         \frac{1}{n!}\,\mathrm{Tr}_{1,\ldots,1+n}\,g_{1+n}(t,1,\ldots,n+1),
\end{eqnarray*}
and, for example, in case of initial data satisfying chaos property (\ref{rozvChaosN}),
if $\| g_1(0)\|_{\mathfrak{L}^{1}(\mathcal{H})}<e^{-1}$, we get
\begin{eqnarray*}
    &&|\langle N\rangle(t)| \leq e\| g_1(0)\|_{\mathfrak{L}^{1}(\mathcal{H})} \sum\limits_{n=0}^{\infty}(n+1)e^n
        (\| g_1(0)\|_{\mathfrak{L}^{1}(\mathcal{H})})^{n}<\infty.
\end{eqnarray*}

\subsection{Marginal density operators}
In view of functional (\ref{averagegs}) we define the marginal density operators $F_s(t)$ 
in terms of correlation operators of clusters of particles 
$g^{(s)}(t)=(g_{1+0}(t,\{Y\}),\ldots,g_{1+n}(t,\{Y\},s+1,\ldots,s+n),\ldots)$ by the expansion
\begin{eqnarray}\label{FClusters}
    &&F_{s}(t,Y)\doteq\sum\limits_{n=0}^{\infty}\frac{1}{n!}\,
       \mathrm{Tr}_{s+1,\ldots,s+n}\,\,g_{1+n}(t,\{Y\},s+1,\ldots,s+n),\quad s\geq1,
\end{eqnarray}
where $g_{1+n}(t,\{Y\},s+1,\ldots,s+n),\,n\geq0$, is  solution (\ref{rozvNF-N_F_clusters}) of the von Neumann hierarchy (\ref{gClusterhierarchy}).
Then the average value functional
of $s$-particle observable (\ref{averagegs}) is given by the formula \cite{BogLect}
\begin{eqnarray*}\label{averageF}
      &&\langle A^{(s)}\rangle(t)=\frac{1}{s!}\,\mathrm{Tr}_{1,\ldots,s}\,a_{s}(1,\ldots,s)F_s(t,1,\ldots,s).
\end{eqnarray*}

The initial-value problem for the marginal ($s$-particle) density operators (\ref{FClusters}) has the following form
(BBGKY hierarchy of the Bose and Fermi many-particle systems) \cite{DP}
\begin{eqnarray} \label{BogolyubovHierarchy}
       &&\frac{d}{dt}F_{s}(t,Y)=-\mathcal{N}_{s}(Y)F_{s}(t,Y)+\nonumber\\
       &&+\sum\limits_{n=1}^{\infty}\frac{1}{n!}
          \mathrm{Tr}_{s+1,\ldots,s+n}
          \sum\limits_{\mbox{\scriptsize $\begin{array}{c}{Z\subset Y},\\Z\neq\emptyset\end{array}$}}
          \big(-\mathcal{N}_{\mathrm{int}}^{(|Z|+n)}\big)(Z,s+1,\ldots,s+n)F_{s+n}(t),\\
    \label{BogolyubovHierarchyInitData}
       &&F_{s}(t)\mid_{t=0}=F_{s}(0),\quad s\geq 1.
\end{eqnarray}

We remark, that as follows from (\ref{averageD}) in the framework of the description of states by the density operators $D=(1,D_{1}(t),\ldots,D_{n}(t),\ldots)$ the marginal density operators $F_s(t),\, s\geq1$, are defined by the well-known formula
(the case of the non-equilibrium grand canonical ensemble \cite{CGP97})
\begin{eqnarray*}\label{F(D)}
   &&F_{s}(t,1,\ldots,s)=(I,D(0))^{-1}
      \sum\limits_{n=0}^{\infty}\frac{1}{n!}\mathrm{Tr}_{s+1,\ldots,s+n}D_{s+n}(t,1,\ldots,s+n),
\end{eqnarray*}
or in terms of definitions (\ref{oper_d}) and (\ref{a}), we rewrite this formula in the following form
\begin{eqnarray*}
   &&F_{|Y|}(t,Y)=(e^\mathfrak{a}D(0))^{-1}_{0}(e^\mathfrak{a}\mathfrak{d}_{Y}D(t))_{0}.
\end{eqnarray*}
Then as above owing to generalized cluster expansions (\ref{gcea})
and the validity of equalities (\ref{d_gamma}) and (\ref{efg}), we deduce expansion (\ref{FClusters})
\begin{eqnarray*}
   &&F_{|Y|}(t,Y)=(e^\mathfrak{a}\,\mathfrak{d}_{\{Y\}}g(t))_{0}.
\end{eqnarray*}
Thus, along with the definition in the framework of the non-equilibrium grand canonical ensemble
the marginal density operators can be defined in the framework of dynamics of correlations
that allows to give the rigorous meaning of the states for more general
classes of operators than trace class operators.

In section 4.4 on basis of the solution of the von Neumann hierarchy (\ref{gClusterhierarchy})
the solution of Cauchy problem (\ref{BogolyubovHierarchy})--(\ref{BogolyubovHierarchyInitData})
of the BBGKY hierarchy for the Bose and Fermi many-particle systems is constructed.

\subsection{Marginal correlation operators}
In view of (\ref{dispg}) we introduce the marginal correlation operators by the series
\begin{eqnarray}\label{Gexpg}
   &&G_{s}(t,1,\ldots,s)\doteq\sum\limits_{n=0}^{\infty}\frac{1}{n!}\,
      \mathrm{Tr}_{s+1,\ldots,s+n}\,\,g_{s+n}(t,1,\ldots,s+n),\quad s\geq1,
\end{eqnarray}
where $g_{s+n}(t,1,\ldots,s+n),\,n\geq0$, is solution (\ref{rozvNF-N_F}) of the von Neumann hierarchy (\ref{vonNeumanHierarchyAlternateDescriptionFB}).
Thus, dispersion functional (\ref{dispg}) represents as follows
\begin{eqnarray*}
   &&\langle(A^{(1)}-\langle A^{(1)}\rangle)^2\rangle(t)=\mathrm{Tr}_1\,\big(a_1^2(1)-\langle A^{(1)}\rangle^2(t)\big)G_1(t,1)+
      \mathrm{Tr}_{1,2}\,a_1(1)a_1(2)\,G_{2}(t,1,2)\nonumber.
\end{eqnarray*}

We note, that the evolution of marginal correlation operators (\ref{Gexpg}) is governed by
the initial-value problem of the nonlinear BBGKY hierarchy \cite{BogLect},\cite{Gre56}.

Marginal correlation operators (\ref{Gexpg}) are interpreted as the operators that describe correlations
of the marginal states (\ref{FClusters}). In fact, it holds
\begin{eqnarray}\label{gBigfromDFB}
    &&G_{s}(t,Y)=\sum\limits_{\mbox{\scriptsize $\begin{array}{c}\mathrm{P}:Y=\bigcup_{i}X_{i}\end{array}$}}
        (-1)^{|\mathrm{P}|-1}(|\mathrm{P}|-1)!\,
        {\mathcal{S}^{\pm}_{n}}\prod_{X_i\subset \mathrm{P}}F_{|X_i|}(t,X_i), \quad s\geq1,
\end{eqnarray}
i.e. the marginal correlation operators are the cumulants (semi-invariants) of the marginal density operators.
In terms of mapping (\ref{circledLn}) sequence (\ref{gBigfromDFB}) can be rewritten in the form
\begin{eqnarray*}\label{gBigfromDFBa}
     && G(t)={\mathbb L}\mathrm{n}_{\circledast}F(t),
\end{eqnarray*}
where $F(t)=\textbf{1}+(0,F_1(t),\ldots, F_n(t),\ldots)$.
As consequence in terms of the marginal density operators
dispersion functional (\ref{dispg}) takes the known form \cite{BogLect}
\begin{eqnarray*}
   &&\langle(A^{(1)}-\langle A^{(1)} \rangle)^2\rangle(t)=\mathrm{Tr}_1\,\big(a_1^2(1)-\langle A^{(1)} \rangle^2(t)\big)F_1(t,1)+\\
   &&+\,\mathrm{Tr}_{1,2}\,a_1(1)a_1(2)\big(F_2(t,1,2)-\mathcal{S}_2^{\pm}F_1(t,1)F_1(t,2)\big)\nonumber.
\end{eqnarray*}

Let us prove the validity of relations (\ref{gBigfromDFB}). Consider expression $(\mathfrak{d}_{Y}\mathbb{E}\mathrm{xp}_{\circledast}G(t))_{0}$.
In view of expansion (\ref{Gexpg}), using definition (\ref{circledExp}) of the mapping ${\mathbb E}\mathrm{xp}_{\circledast}$
and equality (\ref{efg}), we obtain
\begin{eqnarray*}
    &&(\mathfrak{d}_{Y}\mathbb{E}\mathrm{xp}_{\circledast}G(t))_{0}
         =(\mathfrak{d}_{Y}\mathbb{E}\mathrm{xp}_{\circledast}
         e^\mathfrak{a}g(t))_{0}=\\
    &&=\sum\limits_{\mathrm{P}:\,Y=\bigcup_i X_i}
         \mathcal{S}^{\pm}_{|Y|}\,\prod_{X_i\subset\mathrm{P}}(e^\mathfrak{a}\mathfrak{d}_{X_{i}}g(t))_{0}=
         (e^\mathfrak{a}\sum\limits_{\mathrm{P}:\,Y=\bigcup_i X_{i}}\mathfrak{d}_{X_1}g(t)\circledast\ldots
         \circledast \mathfrak{d}_{X_{|\mathrm{P}|}}g(t))_{0}.
\end{eqnarray*}
Then according to equalities (\ref{efg}) and (\ref{oper_dex}), we have
\begin{eqnarray*}
    &&(e^\mathfrak{a}\mathbb{E}\mathrm{xp}_{\circledast}g(t))^{-1}_{0}
         (e^\mathfrak{a}\mathbb{E}\mathrm{xp}_{\circledast} g(t))_{0}
         (e^\mathfrak{a}\sum\limits_{\mathrm{P}:\,Y=\bigcup_i X_{i}}\mathfrak{d}_{X_1}g(t)\circledast\ldots
         \circledast \mathfrak{d}_{X_{|\mathrm{P}|}}g(t))_{0}=\nonumber\\
    &&=(e^\mathfrak{a}\mathbb{E}\mathrm{xp}_{\circledast}g(t))^{-1}_{0}
         (e^\mathfrak{a}(\mathbb{E}\mathrm{xp}_{\circledast} g(t)\circledast
         \sum\limits_{\mathrm{P}:\,Y=\bigcup_i X_{i}}\mathfrak{d}_{X_1}g(t)\circledast
         \ldots\mathfrak{d}_{X_{|\mathrm{P}|}}g(t)))_{0}=\\
    &&=(e^\mathfrak{a}\mathbb{E}\mathrm{xp}_{\circledast}g(t))^{-1}_{0}
         (e^\mathfrak{a}\mathfrak{d}_{Y}\mathbb{E}\mathrm{xp}_{\circledast}g(t))_{0}.
\end{eqnarray*}
Observing that from relations (\ref{gCluster}) it follows that
\begin{eqnarray*}
      &&\mathfrak{d}_{Y}\mathbb{E}\mathrm{xp}_{\circledast}g(t)=\mathfrak{d}_{\{Y\}}\mathbb{E}\mathrm{xp}_{\circledast}g(t),
\end{eqnarray*}
and using equalities (\ref{d_gamma}) and (\ref{efg}), we derive
\begin{eqnarray*}
    &&(e^\mathfrak{a}\mathbb{E}\mathrm{xp}_{\circledast}g(t))^{-1}_{0}
        (e^\mathfrak{a}\mathfrak{d}_{Y}\mathbb{E}\mathrm{xp}_{\circledast}g(t))_{0}=
        (e^\mathfrak{a}\mathbb{E}\mathrm{xp}_{\circledast}g(t))^{-1}_{0}
        (e^\mathfrak{a}\mathfrak{d}_{\{Y\}}\mathbb{E}\mathrm{xp}_{\circledast}g(t))_{0}=\\
    &&=(e^\mathfrak{a}\mathbb{E}\mathrm{xp}_{\circledast}g(t))^{-1}_{0}
        (e^\mathfrak{a}\,\mathbb{E}\mathrm{xp}_{\circledast}\,g(t)\circledast\mathfrak{d}_{\{Y\}}g(t))_{0}
        =(e^\mathfrak{a}\,\mathfrak{d}_{\{Y\}}g(t))_{0}.
\end{eqnarray*}
Therefore, taking into account definition (\ref{FClusters}), we finally establish the validity of the cluster expansions
\begin{eqnarray*}\label{G(g)}
    &&F(t)=\mathbb{E}\mathrm{xp}_{\circledast}G(t).
\end{eqnarray*}
The solutions of obtained recursive relations in component-wise form are expansions (\ref{gBigfromDFB}).

\subsection{On solutions of the BBGKY hierarchies in case of chaos initial data}
Let us apply introduced expansions of marginal correlation operators (\ref{Gexpg}) and
marginal density operators (\ref{FClusters}) to construct solutions of the Cauchy problems of the corresponding BBGKY hierarchies
for one physically motivated example of initial data,
namely initial data satisfying chaos property (\ref{posl_g(0)}).

In terms of marginal correlation operators (\ref{Gexpg}) a chaos property means that
\begin{eqnarray}\label{inG}
   &&G_s(0,1,\ldots,s)=g_1(0,1)\delta_{s,1},\quad s\geq1,
\end{eqnarray}
where $\delta_{s,1}$ is a Kronecker symbol.
Taking into account the structure of solution (\ref{rozvChaosN}) of the initial-value problem of
the von Neumann hierarchy (\ref{vonNeumanHierarchyAlternateDescriptionFB})
in case of initial data (\ref{posl_g(0)}), for expansion (\ref{Gexpg}) we obtain
\begin{eqnarray*}\label{Gg(0)}
   &&G_{s}(t,1,\ldots,s)=\sum\limits_{n=0}^{\infty}\frac{1}{n!}
       \,\mathrm{Tr}_{s+1,\ldots, s+n}\,\mathfrak{A}_{s+n}(t,1,\ldots,s+n)
       \mathcal{S}_{s+n}^{\pm}\,\prod_{i=1}^{s+n}g_{1}(0,i),
\end{eqnarray*}
where $\mathfrak{A}_{s+n}(t)$ is the $(s+n)$-order cumulant (\ref{cumulants}).
In consequence of (\ref{inG}) we finally derive
\begin{eqnarray}\label{GUG(0)}
    &&G_{s}(t,1,\ldots,s)=\sum\limits_{n=0}^{\infty}\frac{1}{n!}
        \,\mathrm{Tr}_{s+1,\ldots, s+n}\,\mathfrak{A}_{s+n}(t,1,\ldots,s+n)
        \mathcal{S}_{s+n}^{\pm}\,\prod_{i=1}^{s+n}G_{1}(0,i),\quad s\geq 1.
\end{eqnarray}
Since the estimate holds
\begin{eqnarray*}
   &&\|\mathfrak{A}_{n}(t)f\|_{\mathfrak{L}^{1}(\mathcal{H}^{\pm}_n)}
        \leq\, n!\,e^n\,\|f\|_{\mathfrak{L}^{1}(\mathcal{H}^{\pm}_n)},
\end{eqnarray*}
series (\ref{GUG(0)}) converges, if $\|G_{1}(0)\|_{\mathfrak{L}^{1}(\mathcal{H})}\leq e^{-1}$.  

Thus, the cumulant structure of solution (\ref{rozvNF-N_F}) of the von Neumann hierarchy
(\ref{vonNeumanHierarchyAlternateDescriptionFB}) induces the cumulant structure of solution expansion (\ref{GUG(0)})
of the initial-value problem of the nonlinear BBGKY hierarchy for marginal correlation operators.

In terms of marginal density operators (\ref{FClusters}) chaos property (\ref{gChaos}) means that
\begin{eqnarray}\label{Fsg1}
    &&F_s(0,Y)=g_1(0,\{Y\}),\quad s\geq 1,
\end{eqnarray}
where $Y=(1,\ldots,s)$ and $\{Y\}$ is a notation introduced in section 2.4.
According to formula (\ref{gtUg0}) for a solution of the initial-value problem of the von Neumann hierarchy (\ref{vonNeumanHierarchyAlternateDescriptionFB}) from definition (\ref{FClusters}) we obtain
\begin{eqnarray*}
  &&F_s(t,Y)=\sum\limits_{n=0}^{\infty}\frac{1}{n!}
       \,\mathrm{Tr}_{s+1,\ldots,s+n}\,\mathfrak{A}_{1+n}(t,\{Y\},X\setminus Y)\,
       \mathcal{S}_{s+n}^{\pm}\, \prod\limits_{i\in (\{Y\},X\setminus Y)}\,g_1(0,i),
\end{eqnarray*}
where $\mathfrak{A}_{1+n}(t)$ is $(n+1)$-order cumulant (\ref{cumcp}).
From definition (\ref{FClusters}) in view of equality (\ref{Fsg1}),
i.e. $g_1(0,i)={\prod\limits}_{j\in \theta(i)}F_1(0,j)$,
we finally derive \cite{DP}
\begin{eqnarray}\label{FUg}
   &&F_{s}(t,Y)=\sum\limits_{n=0}^{\infty}\frac{1}{n!}
       \,\mathrm{Tr}_{s+1,\ldots, s+n}\,\mathfrak{A}_{1+n}(t,\{Y\},X\setminus Y)
       \mathcal{S}_{s+n}^{\pm}\,\prod_{i=1}^{s+n}F_{1}(0,i), \quad s\geq 1.
\end{eqnarray}
We point out that in case of chaos initial data obtained solution expansion (\ref{FUg}) differs from
solution expansion (\ref{GUG(0)}) of marginal correlation operators only by the order
of the cumulants of the groups of operators of the von Neumann equations.

For arbitrary initial data $F(0)\in\mathfrak{L}_{\alpha}^1(\mathcal{F}^{\pm}_{\mathcal{H}})=
\oplus_{n=0}^{\infty}\alpha^n \mathfrak{L}^{1}(\mathcal{H}^{\pm}_n)$,
if $\alpha>e$, for $t\in \mathbb{R}$ there exists a unique solution of the initial value problem (\ref{BogolyubovHierarchy})-(\ref{BogolyubovHierarchyInitData}) given by the series
\begin{eqnarray}\label{FsSolution}
     &&F_s(t,Y)=\sum\limits_{n=0}^{\infty}\frac{1}{n!}\,
          \,\mathrm{Tr}_{s+1,\ldots,s+n}\,\mathfrak{A}_{1+n}(t,\{Y\},X\backslash Y)\,F_{s+n}(0,X),\quad s\geq 1,
\end{eqnarray}
where $\mathfrak{A}_{1+n}(t)$ is $(n+1)$-order cumulant (\ref{cumcp}) of groups of operators (\ref{groupG}).
For initial data $F(0)\in \mathfrak{L}^1_{\alpha,0}(\mathcal{F}^{\pm}_{\mathcal{H}})
\in\mathfrak{L}^1_{\alpha}(\mathcal{F}^{\pm}_{\mathcal{H}})$
it is a strong solution and for arbitrary initial data from the spaces $\mathfrak{L}^1_{\alpha}(\mathcal{F}^{\pm}_{\mathcal{H}})$
it is a weak solution.

It should be emphasized that the solution expansions (\ref{FsSolution}) for the Fermi
and Bose many-particle systems has the same structure, as in case of the Maxwell-Boltzmann statistics \cite{GerS}
which caused by the fact that symmetrization and anti-symmetrization operators (\ref{Sn}) are integrals of motion.

We remark that a solution of initial-value problem (\ref{BogolyubovHierarchy})-(\ref{BogolyubovHierarchyInitData})
is usually represented by the perturbation (iteration) series \cite{AA},\cite{BDGM},\cite{ESY},\cite{Sp80},\cite{M1},\cite{Petr72}.
For certain classes of interaction potentials on suitable subspaces of $\mathfrak{L}_{\alpha}^{1}(\mathcal{F}^{\pm}_\mathcal{H})$
for group (\ref{groupG}) and cumulants (\ref{cumcp}) the analogues of the Duhamel formula hold \cite{BanArl}
and as a result expansion (\ref{FsSolution}) reduces to the iteration series of the BBGKY hierarchy.
For classical systems of particles the first few terms of cumulant solution expansion
for the BBGKY hierarchy were established in \cite{Gre56},\cite{Co68}.

\section{Conclusions}

In the paper it was considered the origin of the microscopic description of
non-equilibrium correlations of Bose and Fermi many-particle systems.
For the correlation operators that give an alternative description of the quantum state evolution of Bose and Fermi many-particle systems,
the von Neumann hierarchy of nonlinear evolution equations (\ref{vonNeumanHierarchyAlternateDescriptionFB})
was introduced.

The properties of group of nonlinear operators generated by the von Neumann hierarchy were examined (Theorem 1).
It gave us the possibility to construct solution (\ref{rozvNF-N_F})
of the Cauchy problem of this hierarchy in the corresponding spaces of sequences of trace class operators (Theorem 2).

In particular, it was established that in case of absence of correlations in the system at initial time
the correlations generated by the dynamics of a system are completely governed by cumulants (\ref{cumulants})
of the groups of operators (\ref{groupG}) of the von Neumann equations.

The links of constructed solution of the von Neumann hierarchy both with the solution of the BBGKY hierarchy
and with the nonlinear BBGKY hierarchy for marginal correlation operators were discussed.
Solution (\ref{FsSolution}) of the BBGKY hierarchy (\ref{BogolyubovHierarchy})
for Fermi and Bose particles with a $n$-body interaction potential was constructed. This solution is represented
in the form of expansion over the clusters of particles, which evolution is governed by the corresponding order
cumulants (\ref{cumcp}) of the groups of operators of the von Neumann equations (\ref{groupG}). The cumulants for the Fermi
and Bose systems of particles have the same structure, as in case of the Maxwell-Boltzmann statistics \cite{GerS}.

The cumulant structure of solution (\ref{rozvNF-N_F}) of the von Neumann hierarchy
(\ref{vonNeumanHierarchyAlternateDescriptionFB}) induces the cumulant structure of solution expansion (\ref{GUG(0)})
both of the initial-value problem of the nonlinear BBGKY hierarchy for marginal correlation operators
and of initial-value problem of the BBGKY hierarchy for marginal density operators. Thus, the dynamics of
infinite-particle systems is generated by the dynamics of correlations.

As we mention above, for initial data from spaces
$\mathfrak{L}^{1}_{\alpha}(\mathcal{F}^{\pm}_\mathcal{H})$  the
average number of particles is finite. In order to describe the evolution of infinitely many particles \cite{CGP97},
in particular, to derive the quantum kinetic equations from constructed dynamics
we have to construct solutions for initial marginal density operators or marginal correlation operators
that belong to more general Banach spaces than $\mathfrak{L}^{1}_{\alpha}(\mathcal{F}^{\pm}_\mathcal{H})$.

\addcontentsline{toc}{section}{References}
\renewcommand{\refname}{References}

\end{document}